\documentclass[conference]{IEEEtran}
\IEEEoverridecommandlockouts
\usepackage{cite}
\usepackage{amsmath,amssymb,amsfonts}
\usepackage{booktabs}
\usepackage{multirow}
\usepackage{algorithmic}
\usepackage{enumitem}
\usepackage{graphicx}
\usepackage{amsthm}     
\newtheorem{theorem}{\textbf{Theorem}}
\usepackage{subcaption}
\graphicspath{{figures/}}
\usepackage{textcomp}
\usepackage{xcolor}
\def\BibTeX{{\rm B\kern-.05em{\sc i\kern-.025em b}\kern-.08em
    T\kern-.1667em\lower.7ex\hbox{E}\kern-.125emX}}
\usepackage[textsize=tiny,textwidth=0.6in]{todonotes}
\DeclareMathOperator{\Var}{Var}
\begin{document}

\title{A First Look at Inter-Cell Interference in the Wild
}

\author{
\IEEEauthorblockN{Daqian Ding}
\IEEEauthorblockA{\textit{Shanghai Jiao Tong University} \\
Shanghai, China\\
daqian.ding@sjtu.edu.cn}
\and
\IEEEauthorblockN{Yibo Pi}
\IEEEauthorblockA{\textit{Shanghai Jiao Tong University} \\
Shanghai, China\\
yibo.pi@sjtu.edu.cn}
\and
\IEEEauthorblockN{Cailian Chen}
\IEEEauthorblockA{\textit{Shanghai Jiao Tong University} \\
Shanghai, China\\
cailianchen@sjtu.edu.cn}
}

\maketitle

\begin{abstract}
In cellular networks, inter-cell interference management has been studied for decades, yet its real-world effectiveness remains under-explored. To bridge this gap, we conduct a first measurement study of inter-cell interference for operational 4G/5G networks. Our findings reveal the prevalence of inter-cell interference and a surprising absence of interference coordination among operational base stations (BSs). As a result, user equipments (UEs) experience unnecessary interference, which causes significant signal quality degradation, especially under frequency-selective channel fading. We examine the inter-cell interference issues from four major perspectives: (1) network deployment, (2) channel assignment, (3) time-frequency resource allocation, and (4) network configuration. In none of these dimensions is inter-cell interference effectively managed. Notably, even when spectrum resources are underutilized and simple strategies could effectively mitigate inter-cell interference, BSs consistently prioritize using the same set of time-frequency resources, causing interference across cells. Our measurements reveal substantial opportunities for improving signal quality by inter-cell interference management.
\end{abstract}

\section{Introduction}

With the global rollout of 5G, the number of deployed 5G base stations (BSs) has exceeded five million worldwide~\cite{5GAmerica2022}, supporting ever-growing demands for mobile services. 5G follows a BS-centric architecture as 4G, where all user equipments (UEs) within a cell connect to the same BS and may experience \emph{inter-cell interference} from neighboring cells using the same frequency bands.
Compared to 4G LTE, 5G operates across a wider range of frequency bands, including mid-band (e.g., 3.5 GHz) and mmWave (e.g., 28 GHz), to provide higher throughput~\cite{3GPP2021}. However, the shift to higher frequencies also results in reduced coverage, demanding denser BS deployments to maintain seamless connectivity. While this densification is essential for achieving the high throughput and low latency promised by 5G, it also exacerbates inter-cell interference, particularly in dense urban environments. Understanding and mitigating inter-cell interference is crucial for optimizing both 4G/5G network performance.

As a fundamental challenge in cellular networks, inter-cell interference has been studied for decades.
A basic solution is to assign different frequencies to neighboring BSs~\cite{subramanian2008minimum}. However, the limited number of available frequency bands is insufficient to support dense deployment, making it essential for frequency reuse. To balance frequency reuse and co-channel interference, each BS's coverage can be divided into inner and outer regions. BSs operate on the same frequency in the inner region and different frequencies in the outer region to minimize inter-cell interference~\cite{hamza2013survey}. To further enhance spectral efficiency, more advanced methods reuse the full spectrum with fine-grained time-frequency coordination among BSs~\cite{yu2013downlink}. 
Decades of research on inter-cell interference management have produced numerous advanced techniques. However, measurement studies on inter-cell interference in operational cellular networks are still lacking. Consequently, the severity of real-world inter-cell interference as well as its potential for improvement remains largely unknown.

\textbf{Contributions.} To bridge this gap, we conduct a first measurement study of inter-cell interference for operational 4G and 5G networks. Our goals are two-fold. (1) We aim to characterize inter-cell interference and its impact on signal quality. Measurement methodology is carefully designed to estimate interference at both coarse and fine-grained granularities.
(2) We aim to mitigate the inter-cell interference from four perspectives, where potential issues are identified and solutions are discussed. Our key findings are summarized below.
\begin{itemize}[leftmargin=*] 
    \item Our study reveals that inter-cell interference is prevalent from both cell and UE perspectives. In 4G and 5G networks, every cell has at least one interfering neighbor, resulting in almost all UEs being affected by inter-cell interference. Due to the denser deployment of 5G, UEs experience interference from significantly more neighboring cells compared to 4G, leading to more severe inter-cell interference. 
    \item We find that both the received signal strength and interference increase with the number of interfering cells per UE. In 4G, interference grows faster than the desired signal strength, leading to a gradual decline in Signal-to-Interference-plus-Noise Ratio (SINR) until the number of interfering cells per UE exceeds a certain threshold. In 5G, this threshold is easily exceeded due to dense deployment, making SINR insensitive to variations in cell density. This indicates that dense deployment in 5G enhances network capacity, but has limited improvement in SINR.
    \item Conducting RB-level channel estimation, we demonstrate the impact of multipath effects on SINR. We observe that the frequency-selective channel leads to significant variations in both channel gain and interference across RBs, with SINR differing by up to 40 dB. This indicates that relying on channel quality measured over a certain frequency range for handover, as in current practice \cite{handover}, can lead to undesired outcomes. This also suggests that resource allocation strategies aware of frequency-selective fading can more effectively utilize spectrum resources.
    \item We find that inter-cell interference can occur between cells on the same BS. About 30\% of 4G BSs and 60\% of 5G BSs suffer from intra-BS inter-cell interference, affecting around 40\% and 70\% of UEs, respectively. Most BSs without such interference are deployed on rooftops, which offer greater spatial flexibility in antenna placement. 
\end{itemize}
To mitigate inter-cell interference, we analyze its root causes and discuss potential solutions from four perspectives: BS deployment, channel allocation, time-frequency resource allocation, and network configuration.
\begin{itemize}[leftmargin=*] 
    \item Due to cost considerations, 5G BS deployment complements the existing 4G infrastructure and has not fully leveraged its advantages of dense deployment.
    \item Both 4G and 5G networks suffer from imbalanced channel allocation across frequency bands. In 4G, most BSs are concentrated on a small subset of channels. In 5G, BSs in the same geographic area are not evenly distributed across channels, worsening inter-cell interference.
    \item Although wireless channel is frequency selective, existing resource allocation adopts a frequency-agnostic strategy that allocates RBs from the lowest frequencies upward. Even when traffic load is low, there is a high probability of resource collisions, leading to inter-cell interference.
    \item Physical Cell ID (PCI) collisions are prevalent in 4G and 5G, causing neighbor cells to transmit RSs on the same time-frequency resources, degrading channel estimation. Moreover, different PCI configurations across signaling message types introduce disparities in their SINRs. As a result, certain messages have low decoding success rates and become bottlenecks in the communication process.
\end{itemize}

\textbf{Dataset and Artifact Release.} We will make our dataset, artifacts, and source code publicly available at a later time.

\section{Measurement Methodology}
\label{sec:methodology}

\begin{figure}[!t]
    \begin{minipage}{0.25\textwidth}
    \centering
    \includegraphics[width=0.8\textwidth]{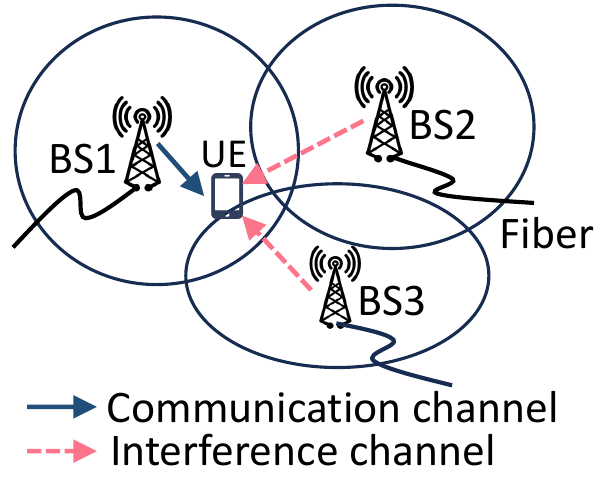}
    \caption{Inter-cell interference}
    \label{fig:inter_cell_interference}
    \end{minipage}
    \hfill
    \begin{minipage}{0.2\textwidth}
    \centering
    \includegraphics[width=0.58\textwidth]{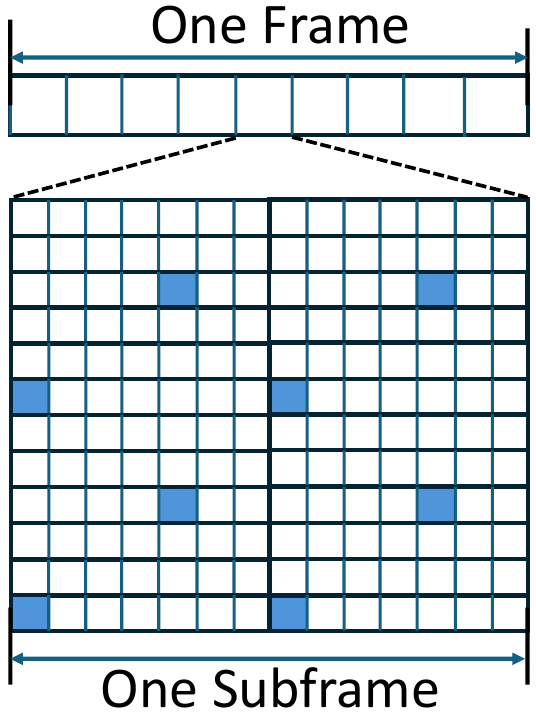}
    \caption{4G RS layout}
    \label{fig:4g_rs}    
    \end{minipage}
    \vspace{-1em}
\end{figure}

We conduct comprehensive measurements within campus region (extended to urban and rural areas in Appendix~\ref{sec:urban}) and carefully design methods for estimating fine-grained RB-level interference, not directly accessible using existing tools.

\textbf{Coverage.} We conduct our measurements on a $2.5km\times1.2km$ campus by traversing all road segments. About 20,840 measurement points are collected in total at 1,389 UE locations. On the campus, we identify 132 4G cells and 197 5G cells as of March 2025, with an average of 53 and 70 measurement points per cell, respectively.
We perform a blanket search of the cells and obtain their information by decoding the Synchronization Signal Block (SSB) and Master Information Block (MIB) following the cell search procedure. The cell search results are further validated using the neighboring cell information decoded from the System Information Blocks (SIBs). Cell-specific metrics are calculated for those with above 30 measurement points to guarantee statistical significance.
As detailed in Table~\ref{tab:band_coverage}, 4G network operates on 7 bands with bandwidths of 10 MHz or 20 MHz, while 5G network operates on 4 bands with bandwidths of 60 MHz or 100 MHz. The 5G center frequencies are higher than those of 4G, while 4G bands have more diverse center frequencies than 5G bands. A band could include more than one channel, each of which is centered at a certain frequency. Our measurement study covers all licensed commercial 4G/5G frequency bands
in our country belonging to three Internet service providers (ISPs), where ISPs 2 and 3 share the same set of frequency bands. 

\textbf{Measurement Tools.} Our measurement campaign was carried out using a smartphone, Xiaomi 10 (Qualcomm X55), and two software-defined radios (USRP B210) that support frequencies from 70 MHz to 6 GHz. The smartphone collects the Reference Signal Received Power (RSRP) and SINR samples every 500ms with CellularZ~\cite{celluarZ}. Meanwhile, the USRPs continuously record the received 4G and 5G OFDM symbols, which are processed to estimate RB-level interference and decoded to retrieve the control channel information of BSs. Note that signaling information (e.g., SSB and SIB) and key performance indicators (e.g., RSRP, SINR) can also be collected using other commercial software, e.g., XCAL\cite{XCAL}, but RB-level channel estimation requires signal processing on the received OFDM symbols.


\textbf{Inter-Cell Interference Detection.} 
We consider two cells \emph{interfering} if their PCIs can be decoded by a UE on the same channel. To identify such cells, we configure the UE to lock onto a specified channel and repeat this for all channels. We extract RRC control messages from the chipset diagnostic log to identify carrier aggregation (CA) configurations and exclude secondary cells assigned via CA from our analysis. We find that CA is rarely used in our measurement region. This allows us to focus only on interfering cells operating on the same channel, not coordinated transmissions by neighbor BSs.

\begin{figure*}[t!]
    \begin{minipage}{0.21\textwidth}
    \centering
    \captionsetup{width=2\textwidth, justification=raggedright, type=table}
    \caption{List of channels covered in measurements}
    \resizebox{1.9\textwidth}{!}{
    \begin{tabular}{cccccc}
    \hline
       & \textbf{Band} & \textbf{Center Freq.} & \textbf{BW} & \textbf{\#Cells} & \textbf{ISP} \\ \hline
    4G & b1            & 2145      & 20      & 5    & ISP2\&3      \\ \cline{2-6} 
       & b3            & 1815        & 20    & 19      & ISP1        \\ \cline{2-6} 
        & b3            & 1835     & 10      & 3      & ISP2\&3      \\ \cline{2-6}
       & b3            & 1850,1870    & 20   & 31,17      & ISP2\&3      \\ \cline{2-6} 
       
       & b8            & 939        & 10    & 20      & ISP1        \\ \cline{2-6} 
       & b8            & 954        & 10    & 2      & ISP2\&3      \\ \cline{2-6} 
        & b34           & 2017.5       & 10     & 6       & ISP1        \\ \cline{2-6} 
       & b39           & 1895      & 20     & 25     & ISP1        \\ \cline{2-6} 
       & b40           & 2330      & 20   & 4    & ISP1        \\   
    \hline 
    
    5G & n41           & 2524.95(2565)     & 100    & 60       & ISP1        \\ \cline{2-6} 
       & n41           & 2662.95(2644.80)   & 60   & 44       & ISP1        \\ \cline{2-6} 
       & n78           & 3408.96,3509.76    & 100   & 42,11         & ISP2\&3      \\ \cline{2-6} 
       & n79           & 4850.4           & 100    & 40        & ISP1        \\ 
    \hline
    \end{tabular}
    }
    \label{tab:band_coverage}        
    \end{minipage}
    \hfill
    \begin{minipage}[!t]{0.57\textwidth}
    \captionsetup{type=figure}
    \begin{subfigure}{0.31\textwidth} 
        \centering
        \includegraphics[width=\textwidth]{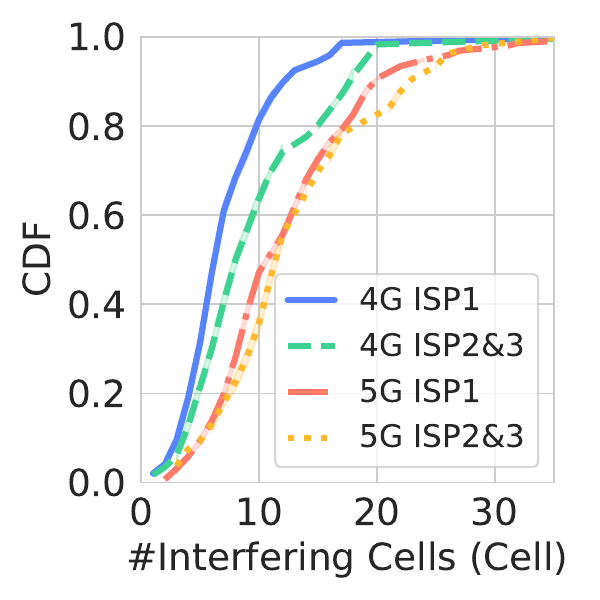}
        \caption{Cell perspective}
        \label{fig:neighbor_cells_cdf_by_isp}
    \end{subfigure}
    \hspace{-0.1cm}
    \begin{subfigure}{0.35\textwidth} 
        \centering
        \includegraphics[width=\textwidth]{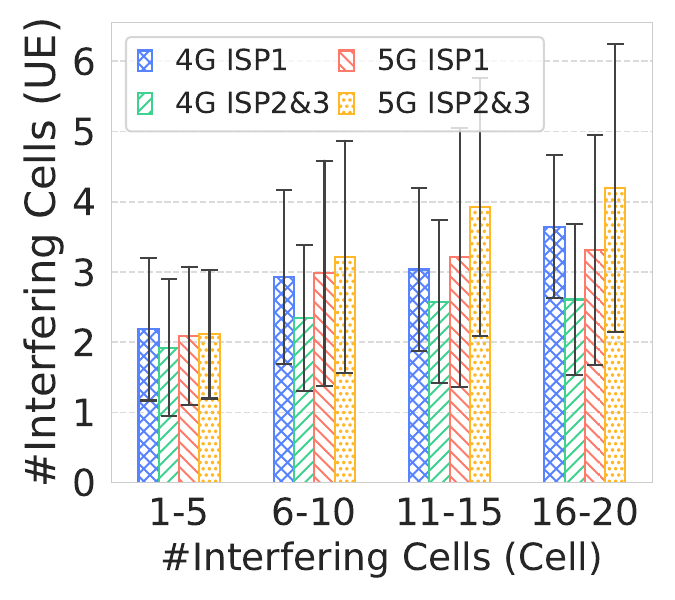}
        \caption{Cell vs. user perspective}
        \label{fig:interference_cell_vs_neighboring_cell_by_isp}
    \end{subfigure}
    \hspace{-0.1cm}
    \begin{subfigure}{0.31\textwidth} 
        \centering
        \includegraphics[width=\textwidth]{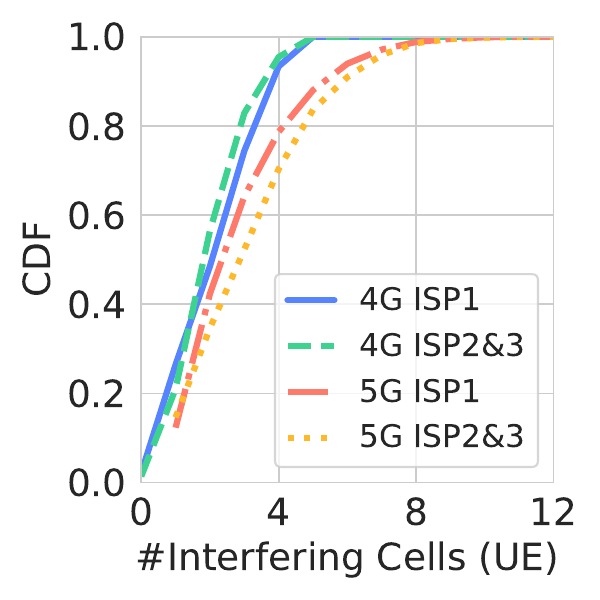}
        \caption{User perspective}
        \label{fig:interfering_cells_cdf_by_isp}
    \end{subfigure}
    \caption{Prevalence of inter-cell interference from the cell and UE perspectives. All UEs are interfered by at least one cell. The average number of interfering cells per UE in 5G is almost twice that in 4G.}
    \label{fig:prevalence}
    \end{minipage}
\end{figure*}

\textbf{UE-Level Interference Estimation}. Per the definition of SINR, interference can be expressed as 
\begin{equation}
    \text{Interference (dBm)} = \text{RSRP (dBm)} - \text{SINR (dB)},
\end{equation}
where RSRP is the average reference signal received power. Since UE-level RSRP and SINR are continously estimated to assess channel quality for handovers between BSs, they are readily accessible on smartphones~\cite{rsrq_handover}.

\textbf{RB-Level Interference Estimation.} Since smartphones do not expose RB-level channel information, we analyze physical-layer received signals to estimate RB-level interference.
Figure~\ref{fig:inter_cell_interference} illustrates a three-node network, where BS 1 serves the UE, and BSs 2 and 3 are interferers on the same channel. Since 4G and 5G employ orthogonal frequency-division multiplexing (OFDM)-based modulation, the signal received on a resource element (RE)\footnote{A RE is the smallest unit of time-frequency resource consisting of one subcarrier in the frequency domain and one OFDM symbol in the time domain.} (denoted by \( Y \)) can be written as
\begin{equation}
    Y = h_1 X_1 + I + N,
\end{equation}
where $X_1$ is the transmitted symbol from BS 1, $h_1$ is the corresponding channel gain, $I$ is the combined interference from BSs 2 and 3, and $N$ is the noise. To estimate the interference $I$, we first estimate the channel gain $h_1$ using the Least Squares (LS) estimator as
\begin{equation}
    \hat{h}_1 = \frac{Y}{X_1} = h_1 + \Delta,
\end{equation}
where $X_1$ is the reference signal from BS 1 known to the UE and $\Delta = (I + N)/X_1$. 
Due to the randomness in interference and noise, the channel estimation accuracy could be significantly affected. To better estimate $h_1$, we combine multiple reference signal REs within coherence time\footnote{Coherence time is the duration over which a wireless channel remains stable, typically on the order of tens of milliseconds for slow-moving objects.} to smooth out interference and noise, which are typically random variables with a zero mean. Let $Y_i$ be the $i$-th OFDM sample for the received signal and the corresponding estimation error is $\Delta_i$. By taking $m$ samples, we can control the average interference and noise within a threshold, i.e., 
\begin{equation}\label{eq:constraint}
\frac{\left|\frac{1}{m}\sum_{i=1}^{m} \Delta_i\right|}{|h_1|} \leq \delta.    
\end{equation}
With a proper $m$, we can estimate $h_1$ as $\hat{h}_1 = \frac{1}{m}\sum_{i=1}^m \frac{Y_i}{X_1}$ and estimate RE-level interference by subtracting the estimated signal, \( \hat{h}_1X_1 \), from the received signal as
\begin{equation}
    I = \frac{1}{m}\sum_{i=1}^m |Y_i - \hat{h}_1X_1|^2.
\end{equation}
RB-level interference is the sum of interference across 12 REs.

Suppose that the UE has $K-1$ interferers indexed from $2$ to $K$. We have the following theorem to choose a proper $m$.
\begin{theorem}\label{theorem:bennet}
When noise is negligible compared to interference, the number of reference signal REs required to enforce Equation~(\ref{eq:constraint}) with the probability of at least $1 - \epsilon$ is
\begin{equation*}
    m \geq \frac{-b^2 \log \epsilon }{\sigma^2}h^{-1}\left(\frac{b\delta}{\sigma^2}\right)
\end{equation*}
where $h^{-1}(t) = 1/((1+t)\log(1+t) - t)$, $\sigma^2 = \frac{1}{SINR}$, and $b = \max\left(\sqrt{\frac{(K-1)\sum_{j=2}^K |h_jX_j|^2}{|h_1X_1|^2}}\right)$. 
\end{theorem}

\begin{proof}
    Please see Appendix~\ref{appendix:theorem} for detailed proof.
\end{proof}






In our measurement campaign, as almost all the UEs have fewer than 6 interfering BSs, we set $K = 7$. We observe that $b$ can be interpreted as $\max\left(\sqrt{\frac{K-1}{\text{RE-level SINR}}}\right)$, where the RE-level SINR is the SINR for single reference signal REs. To determine a practical $b$, we judiciously set the RE-level SINR to $-6.7$dB to support the lowest Modulation and Coding Scheme (MCS) in 5G \cite{mcs_table}. Further, we require that the communication channel gain is 10dB larger than average interference and noise for channel estimation, i.e., $\delta = 0.1$. Let $\epsilon = 0.1$. When SINR is $-6$dB (worse than 99.4\% of samples in our measurements), we need $m = 228$ reference signal REs to provide statistical guarantees for Equation~(\ref{eq:constraint}) to hold. 

To estimate channel gains at the UE side, we use periodic cell-specific RSs (CRSs) in 4G and Synchronization Signal Blocks (SSBs) in 5G. 
One 4G frame consists of 10 subframes, each including 8 cell-specific RSs per 12 subcarriers (Figure~\ref{fig:4g_rs}), whereas one SSB in 5G spans 240 subcarriers and lasts for 4 OFDM symbols. To provide statistical guarantees, we aggregate cell-specific RSs across every 12 subcarriers in three consecutive frames for 4G, totaling 240 REs, and use all REs within the SSB block for 5G, totalling 398 REs. In 4G, three frames last for 30ms, less than the typical channel coherence time for walking UEs.



\section{Characterizing Inter-Cell Interference}

\subsection{Prevalence}
We present the inter-cell interference in 4G and 5G networks from both the cell and UE perspectives. 
\begin{figure*}[t]
    \centering
    \begin{subfigure}[b]{0.32\textwidth}
        \includegraphics[width=\textwidth]{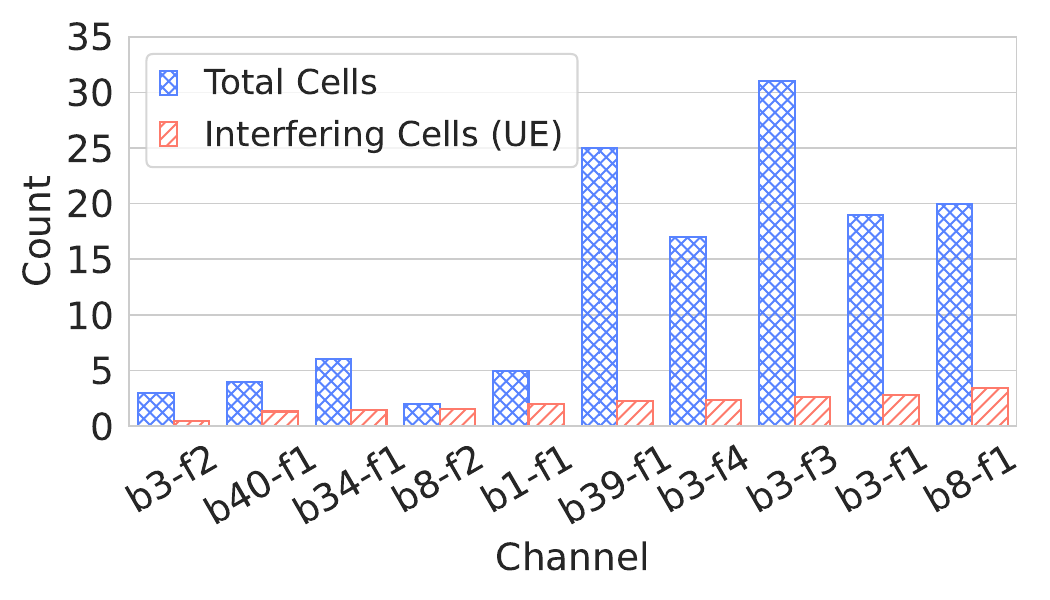}
        \caption{Channels ordered by interfering cell count}
        \label{fig:lte_interference}
    \end{subfigure}
    \hfill
    \begin{subfigure}[b]{0.32\textwidth}
        \includegraphics[width=\textwidth]{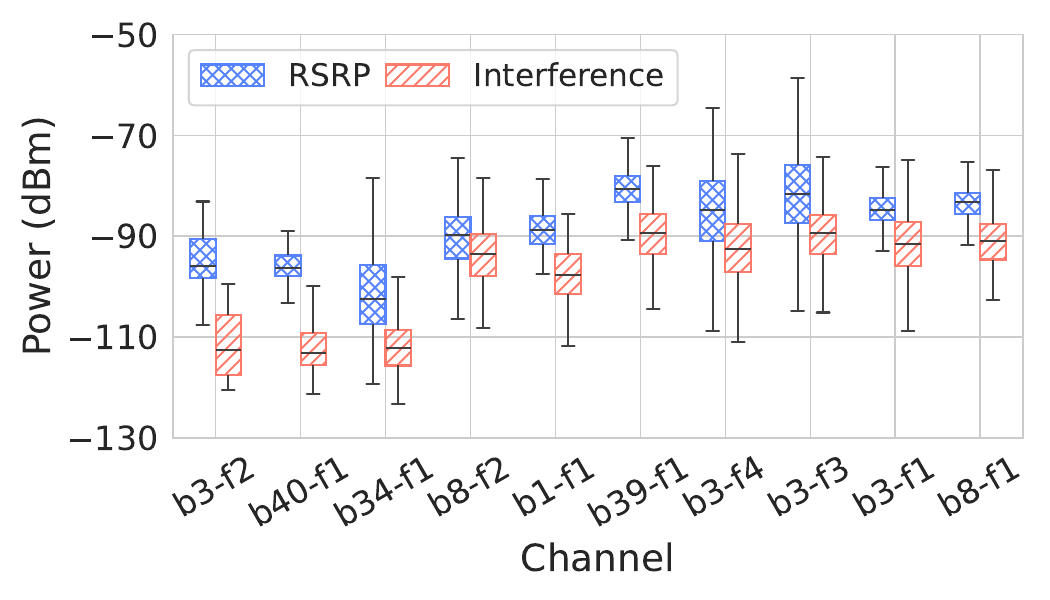}
        \caption{4G RSRP and interference}
        \label{fig:lte_rsrp}
    \end{subfigure}
    \hfill
    \begin{subfigure}[b]{0.32\textwidth}
        \includegraphics[width=\textwidth]{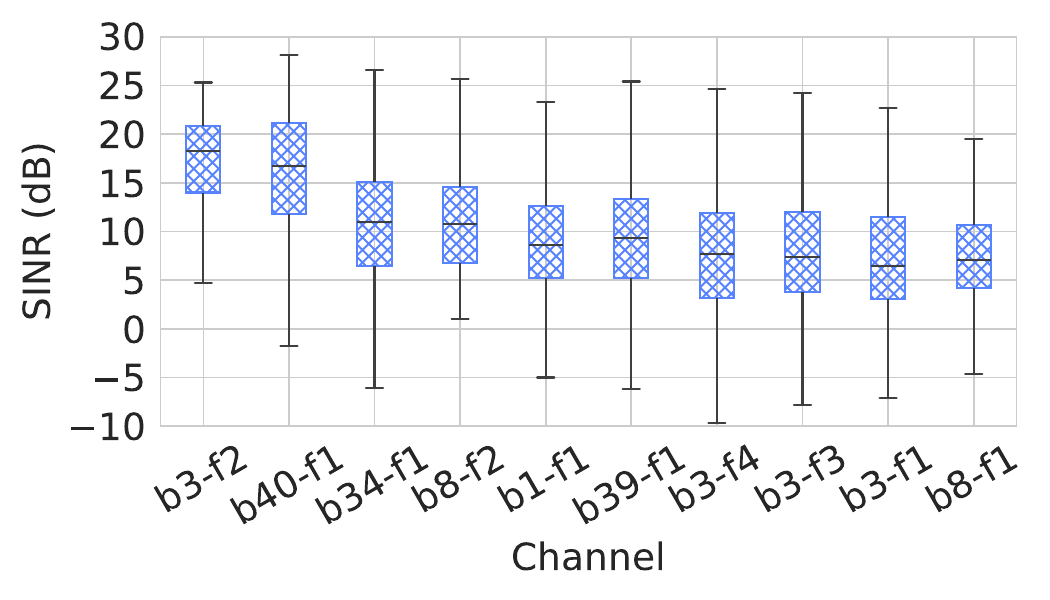}
        \caption{4G SINR}
        \label{fig:lte_sinr}
    \end{subfigure}
    \caption{4G performance metrics ranked by increasing interfering cell count. Both RSRP and interference increase when network deployment is denser. However, as interference grows at a faster rate, the resulting SINR gradually declines.}
    \label{fig:lte_metrics}
    \vspace{-1em}
\end{figure*}

\begin{figure}[t]
    \centering

    \begin{subfigure}[t]{0.15\textwidth}
        \includegraphics[width=\textwidth]{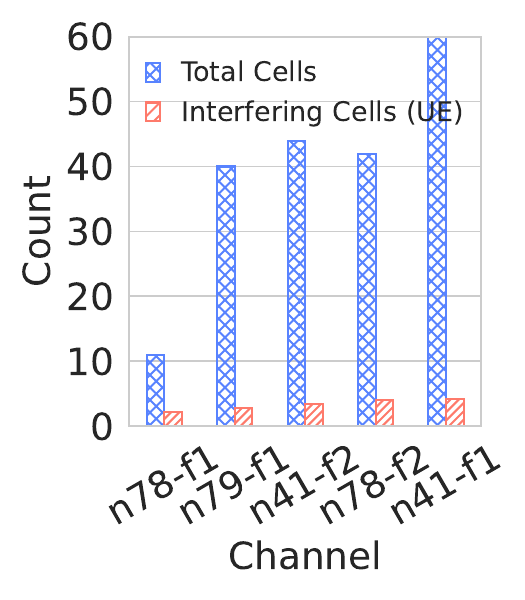}
        \caption{Interfering cell count by channels}
        \label{fig:interference}
    \end{subfigure}
    \hspace{0cm}
    \begin{subfigure}[t]{0.17\textwidth}
        \includegraphics[width=\textwidth]{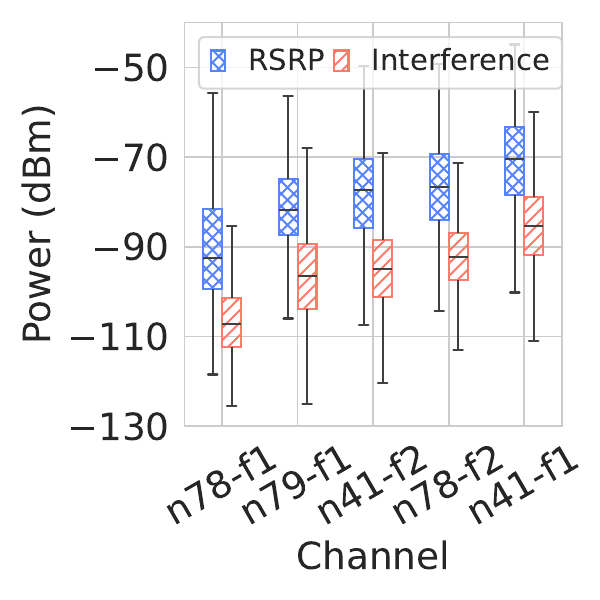}
        \caption{RSRP and interference power}
        \label{fig:5g_rsrp}
    \end{subfigure}
    \hspace{-0.2cm}
    \begin{subfigure}[t]{0.13\textwidth}
        \includegraphics[width=\textwidth]{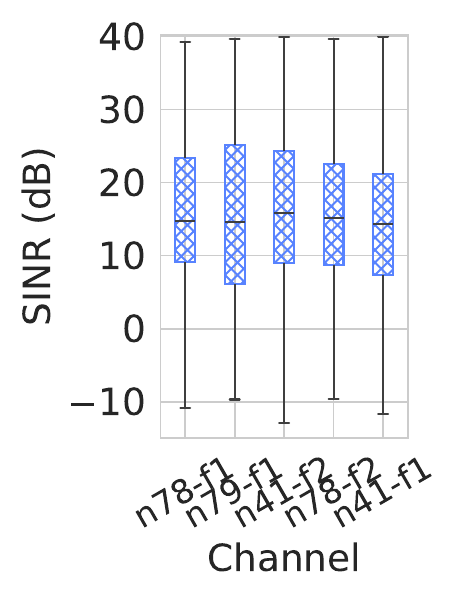}
        \caption{SINR}
        \label{fig:sinr}
    \end{subfigure}
    \caption{In 5G, both RSRP and interference increase with denser deployments. However, since they grow at similar rates, SINR remains relatively stable.}
    \label{fig:metrics}
    \vspace{-1.5em}
\end{figure}

\textbf{Cell Perspective.} 
For each cell, its total interfering cells are determined by aggregating all the interfering cells observed by its connected UEs. Since CA is excluded, each UE connects to only one BS at any given location and all other BSs on the same channel are interferers. As shown in Figure~\ref{fig:neighbor_cells_cdf_by_isp}, inter-cell interference is prevalent in both 4G and 5G networks, with all cells having at least one interfering neighbor. In 4G, over 30\% of cells have more than 10 interfering neighbors, whereas in 5G, this percentage increases to 50\%, primarily due to the denser BS deployment required for 5G. Among ISPs, ISP-1 deploys 2.5 times more 5G cells than ISPs 2 and 3 combined; however, their average cell densities per frequency band are comparable. This is because ISP-1 distributes its 5G cells across 3 bands, while ISPs 2 and 3 operate on a single shared band consisting of two 100 MHz channels (see Table~\ref{tab:band_coverage}).

\textbf{User Perspective.} The number of interfering neighbors of a cell indicates the likelihood and severity of inter-cell interference experienced by its served UEs. As shown in Figure~\ref{fig:interference_cell_vs_neighboring_cell_by_isp}, the number of interfering cells of a UE increases with the number of interfering cells observed by its serving cell in both 4G and 5G. Each UE may be affected by 25\% to 40\% of the interfering cells of its serving cell. Since more than half of the cells have at least 10 interfering neighbors in 5G (Figure~\ref{fig:neighbor_cells_cdf_by_isp}), their served UEs will experience interference from approximately 3 to 4 neighboring cells. As shown in Figure~\ref{fig:interfering_cells_cdf_by_isp}, about 40\% of ISP-1's UEs have more than 3 and 4 interfering cells in 4G and 5G, respectively. Overall, 5G UEs have a higher number of interfering cells compared to 4G UEs. 

\subsection{Impact on Channel Quality}
To analyze the impact of interference on channel quality, we evaluate SINR at both the UE and RB levels.


\textbf{UE-Level SINR.} 
As SINR depends on the signal strength and interference, we examine both impacting factors to characterize SINR, where the signal strength is measured by RSRP. Specifically, we rank the 4G channels based on the average number of interfering cells observed by UEs on each channel (Figure~\ref{fig:lte_interference}) and rank RSRP and interference across channels in the same order (Figure~\ref{fig:lte_rsrp}), where b3-f1 represents the 1st channel in frequency band b3.
We see that both RSRP and interference increase with the interfering cell count due to denser deployment. However, as the interference increases faster, the resulting SINR gradually declines.
In 4G, densification reduces rather than improves SINR. The difference in mean SINR across channels can be higher than 10 dB. The two channels (b3-f1 and b8-f1) with the lowest SINRs do not have the highest total number of cells, indicating a potential deployment issue (see \S\ref{sec:deployment} for details).


Similarly, we present the RSRP and interference across channels ranked by the average number of interfering cells per UE for 5G (Figure~\ref{fig:5g_rsrp}). Both metrics increase with cell density as in 4G. The key difference is that in 5G, RSRP increases at a rate comparable to interference, resulting in relatively stable SINR across channels. This is because even the least interfered channel in 5G has over two interfering cells per UE. When the number of interfering cells per UE exceeds a threshold, the increase rates for RSRP and interference become comparable and the SINR variations remain stable. Comparing 4G, 5G achieves 5 to 7 dB higher SINR due to stronger signal strength, yielding 3$-$5$\times$ increase in channel capacity in the high-SINR regime. Nonetheless, achieving this gain in 5G requires a 2.5$\times$ denser deployment than in 4G. As 5G is still evolving, enhancing interference management offers a more cost-effective approach to further improving channel quality (see \S\ref{sec:resource_allocation} for details).

\textbf{RB-Level SINR.} The multipath effects in wireless channels leads to frequency-selective fading, where channel gains vary across frequencies.
Due to the independence between communication and interference channels, frequencies with weak communication channel gains may experience strong interference, resulting in significantly low RB-level SINRs. As RB-level SINR is easier to measure in 4G due to its dense and periodic RS patterns, our measurement focuses on 4G, but the conclusions generalize to 5G. 
Figure~\ref{fig:freq_sel} illustrates the RB-level frequency-selective channel gains over 3,500 OFDM symbols measured at a random UE location, which reveals varying gains across frequencies. This pattern typically persists over the channel coherence time, which decreases with UE mobility. Both communication and interference channels may experience frequency-selective fading. Extreme RB-level SINRs occur when a weak communication channel gain collides with a strong interference channel gain on certain RBs.
\begin{figure}[!t]
    \centering
    \includegraphics[scale=0.25]{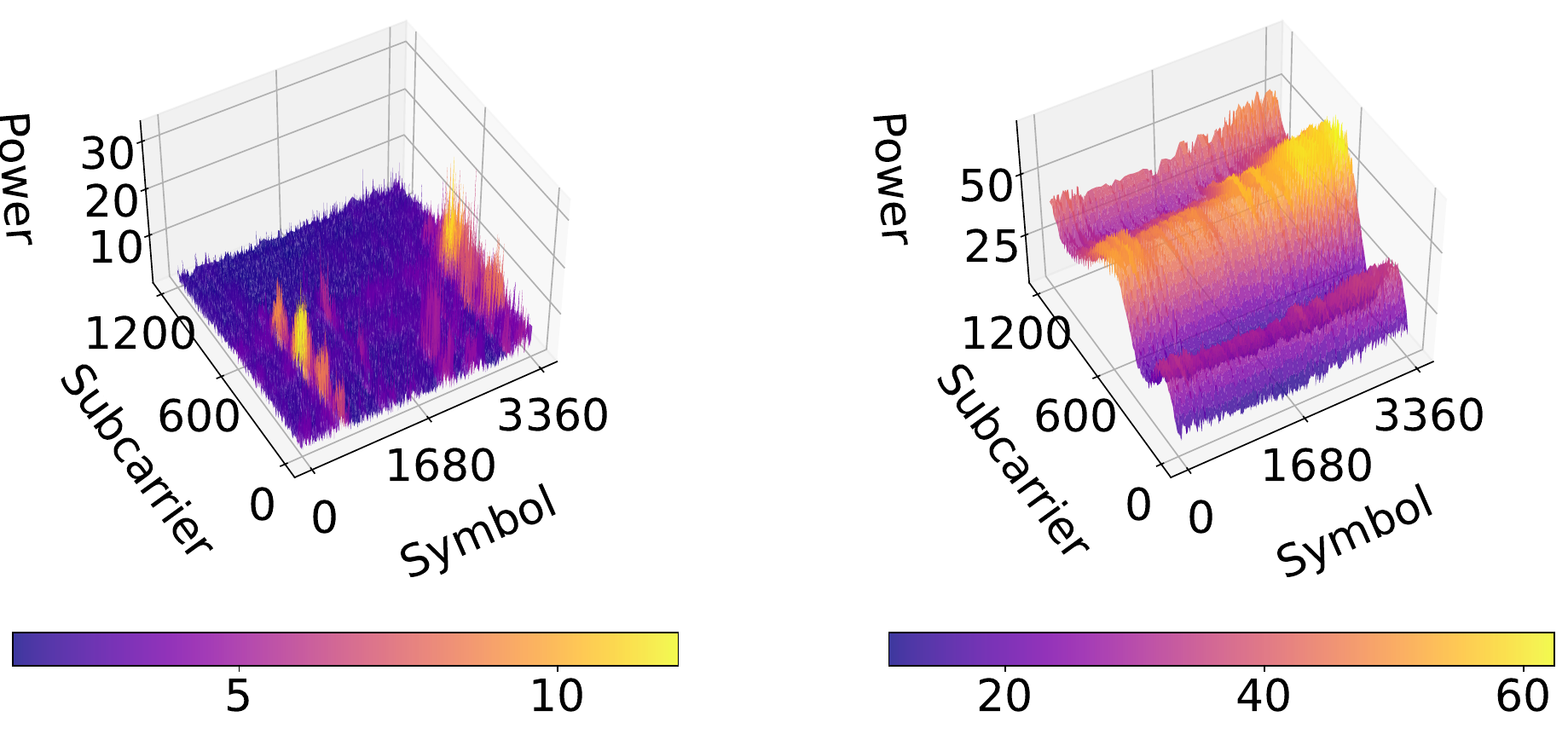}
    \caption{RB-level interference (left) and RSRP (right) in 4G measured at a random UE location. Due to frequency-selective fading, channel gains vary across frequencies.}
    \label{fig:freq_sel}
\end{figure}

\begin{figure}[t]
    \centering
    \begin{subfigure}[t]{0.23\textwidth}
    \includegraphics[width=\textwidth]{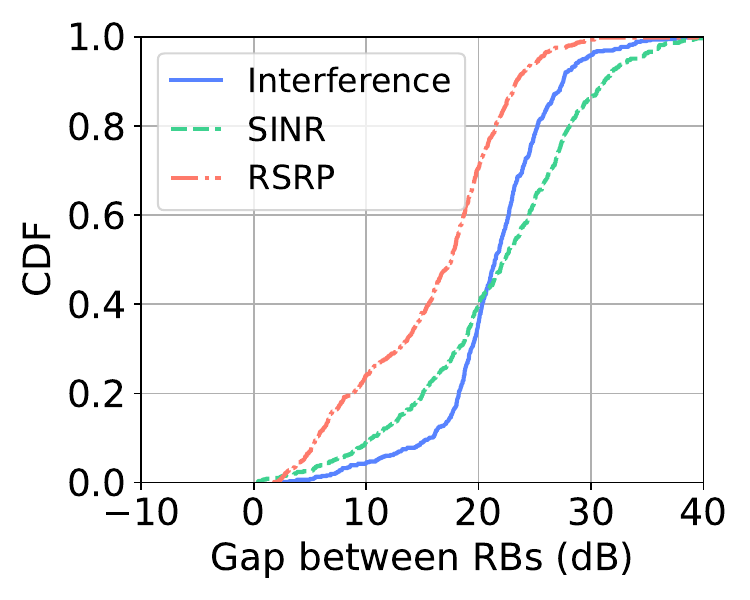}
        \caption{SINR gap between RBs}
        \label{fig:sinr_across_rbs}
    \end{subfigure}
    \hspace{0cm}
    \begin{subfigure}[t]{0.23\textwidth}
        \includegraphics[width=\textwidth]{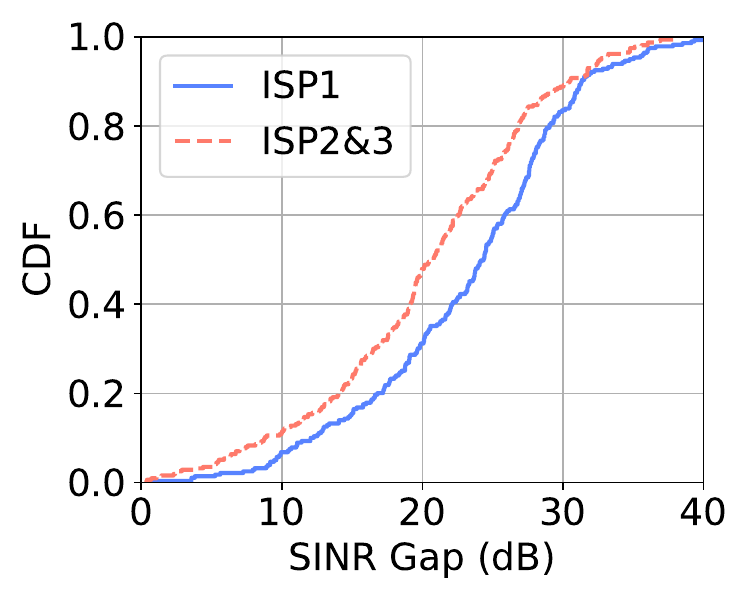}
        \caption{SINR gaps by ISPs}
        \label{fig:SINR_isp_mean_min_cdf}
    \end{subfigure}
    \caption{Distribution of SINR gaps for UEs. Frequency-selective fading leads to varying interference and RSRP gaps at different UE locations. Extremely low RB-level SINRs occur when strong interference aligns with weak desired signals.}
    \vspace{-1em}
\end{figure}

To characterize the impact of frequency-selective fading, we define the \emph{SINR gap across RBs} as the difference between the maximum and minimum RB-level SINRs. Similarly, we compute the RSRP and interference gaps for each UE location. Our measurements show a wide spread of both interference and RSRP gaps, with over 50\% of UEs experiencing gaps greater than 15 dB (Figure~\ref{fig:sinr_across_rbs}). As SINR is the ratio of RSRP to interference, it has a broader distribution, with more than 10\% of UEs showing SINR gaps exceeding 30 dB, and some exceeding 40 dB. UEs across different ISPs show similar SINR gaps, though ISP1 has slightly more extreme SINR gaps (Figure~\ref{fig:SINR_isp_mean_min_cdf}). The impact is also observed across different channels, where channels with higher interference tend to have lower mean and minimum RB-level SINR values (not shown).


The significant variations in RSRP, interference, and SINR across RBs have two key implications: (1) performing time-frequency resource allocation based on per-RB SINR can more effectively utilize spectrum resources and improve user channel quality; (2) handover decisions based on UE-level metrics or limited frequency observations often result in suboptimal outcomes. This limitation is evident in existing 4G and 5G networks, where handovers are typically based on RSRP measured at the UE level and within the frequency range of the SSB, respectively.



\begin{figure}[!t]
    \centering
    \begin{subfigure}{0.22\textwidth}
        \centering
        \includegraphics[scale=0.29]{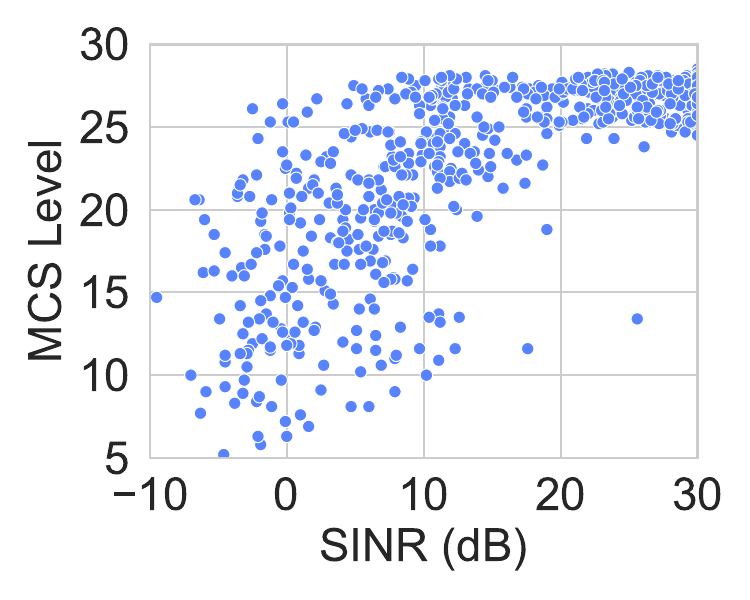}
        \caption{MCS vs. SINR}
        \label{fig:sinr_mcs}
    \end{subfigure} 
    \hfill
    \begin{subfigure}{0.22\textwidth}
        \centering
        \includegraphics[scale=0.29]{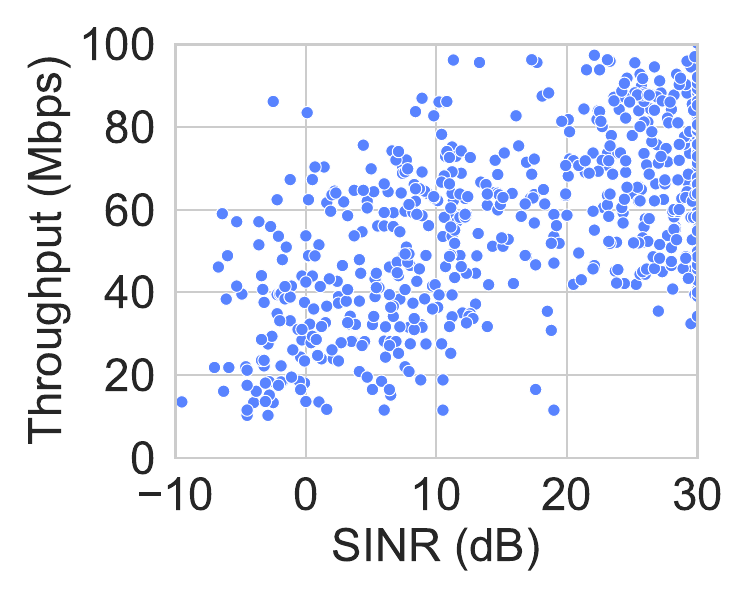}
        \caption{Throughput vs. SINR}
        \label{fig:sinr_throughput}
    \end{subfigure}
    \caption{Impact of SINR on throughput}
\end{figure}

\textbf{Impact of SINR on Throughput.}
We collect measurements of MCS and throughput under different SINRs in 4G using a smartphone. The results show that, before reaching the MCS ceiling, both the average MCS and throughput increase approximately linearly with SINR (in dB). From a network-wide perspective, the average SINR difference across 4G channels can reach up to 10 dB (Figure~\ref{fig:lte_sinr}), resulting in a throughput gap of about 20 Mbps (Figure~\ref{fig:sinr_throughput}). Moreover, RB-level resource allocation can significantly improve user SINR. In 4G, the gap between the best and worst RBs can exceed 40 dB in SINR, spanning more than 30 MCS levels (Figure~\ref{fig:sinr_mcs}).

\begin{figure}[t]
    \begin{minipage}[b]{0.24\textwidth}
    \centering
    \includegraphics[scale=0.3]{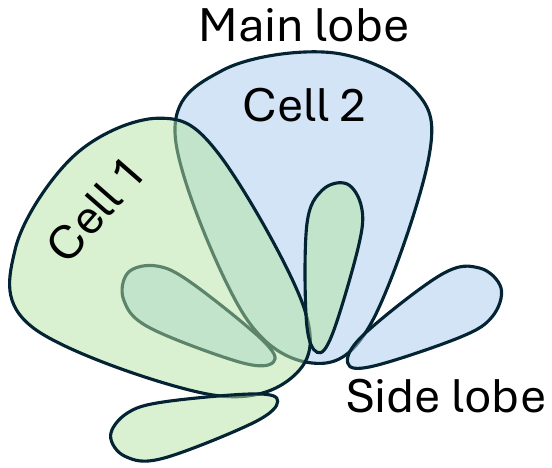}
    \caption{Intra-BS ICI}
    \label{fig:inter_beam_interference}
    \end{minipage}
    \hfill
    \begin{minipage}[b]{0.24\textwidth}
    \centering
    \resizebox{0.9\textwidth}{!}{
    \begin{tabular}{cccc}
    \toprule
    View & ISP & 4G & 5G  \\ 
    \midrule
    \multirow{2}{*}{BS-level} & 1 & 0.31 & 0.60\\
    \cline{2-4}
                          & 2\&3 & 0.29 & 0.62\\
    \cline{2-4}
                          & Total & 0.30& 0.61\\
    \hline
    \multirow{2}{*}{UE-level} & 1 & 0.36 &0.69 \\
    \cline{2-4}
                        & 2\&3 & 0.40 & 0.66\\
    \cline{2-4}
                        & Total & 0.39& 0.68\\
    \bottomrule
    \end{tabular}
    }
    \captionsetup{type=table} 
    \caption{Intra-BS ICI ratio}
    \label{fig:intra_vs_inter}
    \end{minipage}
\end{figure}

\subsection{Dissecting Inter-Cell Interference}
A physical BS can comprise multiple sectorized cells, each covering a different direction. Interference between cells from the same BS is termed \emph{intra-BS inter-cell interference (intra-BS ICI)}, while interference between cells from different BSs is termed \emph{inter-BS ICI}. 

\textbf{Intra- vs. Inter-BS ICI.}
As shown in Figure~\ref{fig:inter_beam_interference}, UEs in cell 1 may experience interference from both the main and side lobes of cell 2. To ensure complete signal coverage, neighboring beams typically have some overlap, but excessive overlap can introduce significant interference and degrade signal quality. In our measurement, we find that cells belonging to the same BS commonly have consecutive PCIs and use this observation to determine cells within the same BS. As shown in Table~\ref{fig:intra_vs_inter}, 30\% and 61\% of BSs in 4G and 5G include interfering cells, respectively. Consequently, for ISP-1, 36\% of 4G UEs experience intra-BS ICI, while 40\% of 5G UEs experience it. 
We notice that BSs not having interfering cells are typically rooftop-mounted. Unlike tower deployments, rooftop installations provide greater flexibility in antenna panel placement, allowing for increased separation and reduced interference.


\begin{figure*}[t]
    \centering
    \begin{subfigure}{0.3\textwidth}
        \centering
        \includegraphics[scale=0.35]{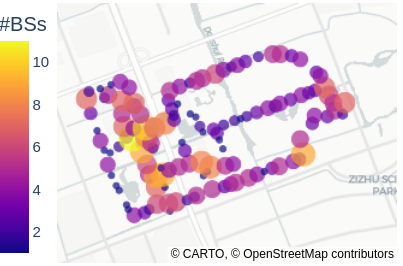}
        \caption{Campus-scale 4G BS density}
        \label{fig:bs_geo_4g}
    \end{subfigure}
    \hfill
    \begin{subfigure}{0.3\textwidth}
        \centering
        \includegraphics[scale=0.35]{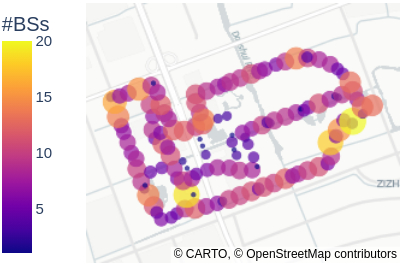}
        \caption{Campus-scale 5G BS density}
        \label{fig:bs_geo_5g}
    \end{subfigure}
    \hfill
    \begin{subfigure}{0.3\textwidth}
        \centering
        \includegraphics[scale=0.33]{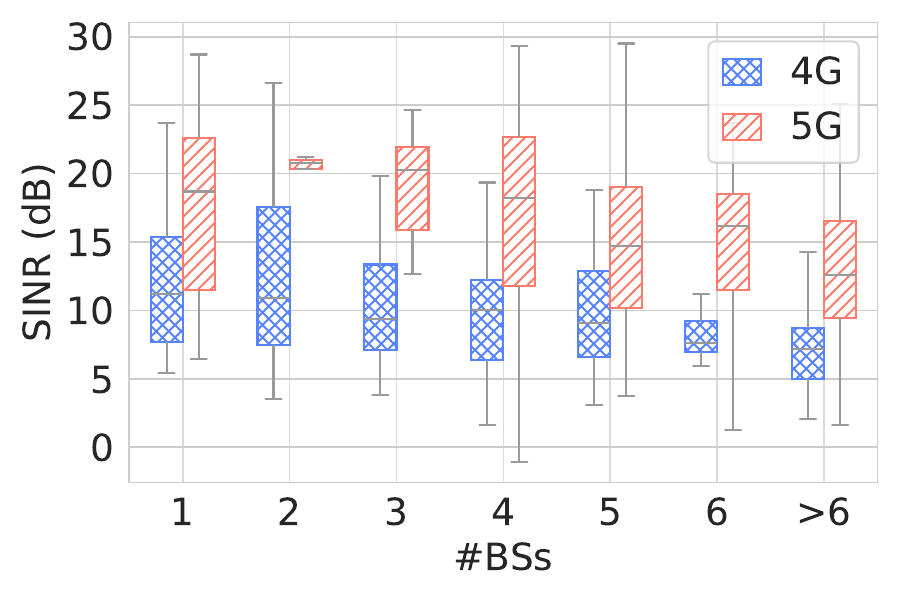}
        \caption{SINR vs. BS density}
        \label{fig:sinr_vs_bs_density}
    \end{subfigure}
    \caption{Geographical density distribution of 4G/5G BSs and its impact on SINR (from UE perspective).}
    \vspace{-0.5em}
\end{figure*}

\textbf{Impact of Intra-BS ICI.} To understand the impact of intra-BS ICI, for UEs experiencing intra-BS ICI, we compute the difference in RSRP between the serving cell and the interfering cell from the same BS, termed \emph{Delta RSRP}. In our measurements, about 30\% of UEs have a Delta RSRP below 6 dB in both 4G and 5G, indicating strong intra-BS ICI. These UEs experience average SINRs at least 5 dB lower than those with no or minimal intra-BS ICI.

\section{Mitigating Inter-Cell Interference}
This section discusses the potential issues and solutions for inter-cell interference from four perspectives.

\subsection{Deployment}
\label{sec:deployment}

BS deployment involves complex trade-offs across multiple dimensions, including deployment cost, geographical constraints, user traffic demand, channel quality and coverage. We discuss two deployment issues with potential improvements.

\textbf{Issues due to Compatibility with 4G.} 
Due to cost constraints, 5G BS deployment typically takes into account existing 4G infrastructure. Prior to 5G, 4G addresses high traffic demand through dense deployment. However, due to the large interference range of 4G, network capacity improvement often comes at the cost of degraded signal quality. Considering the deployment cost, 5G deployment must maintain compatibility with densely deployed 4G infrastructure, which slows down the upgrade process and prevents 5G from realizing its full potential in high-density areas. To better understand the geographical distribution of BSs, we divide the campus area evenly into 100$\times$100 m$^2$ grids. Figure~\ref{fig:bs_geo_4g} and~\ref{fig:bs_geo_5g} show the total number of cells in each grid across the campus for 4G and 5G, respectively.
There is a clear difference in cell density on both sides of the boundary between the west and east campuses for 4G. The dense areas in the west campus are primarily students' dorm and classroom areas, while most of the east campus are spacious areas for laboratories and offices. This indicates that the BS deployment density is related to the population density. The dense deployment of 4G BSs aims to support higher traffic demands. 5G network also exhibits high BS density in the west campus, but its top high-density regions differ from those of 4G. Given its later deployment, 5G network may have been strategically provisioned to complement existing 4G coverage. However, the dense deployment of 4G BSs results in a significant SINR decrease, while the same density of 5G BS deployment boosts the SINR by 10 dB on average (Figure~\ref{fig:sinr_vs_bs_density}).



\subsection{Channel Assignment} 
\begin{figure}[t]
    \centering
    \begin{subfigure}{0.15\textwidth}
        \centering
        \includegraphics[width=\textwidth]{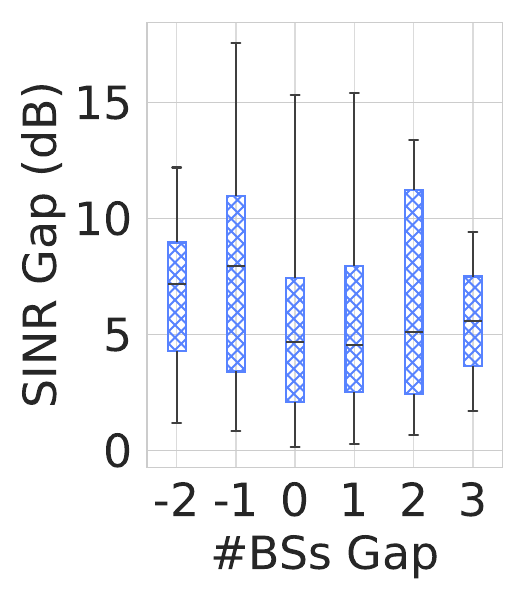}
        \caption{Gap between \\n79-f1 and n41-f1} 
        \label{fig:bs_sinr_gap_isp1}
    \end{subfigure}
    \hspace{0cm}
    \begin{subfigure}{0.15\textwidth}
        \centering
        \includegraphics[width=\textwidth]{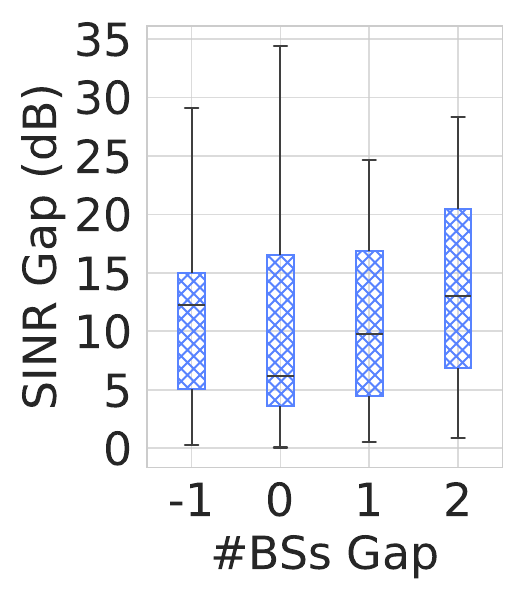}
        \caption{Gap between\\ n78-f2 and n78-f1} 
        \label{fig:bs_sinr_gap_isp2}
    \end{subfigure}
    \hspace{0cm}
    \begin{subfigure}{0.15\textwidth}
        \centering
        \includegraphics[width=\textwidth]{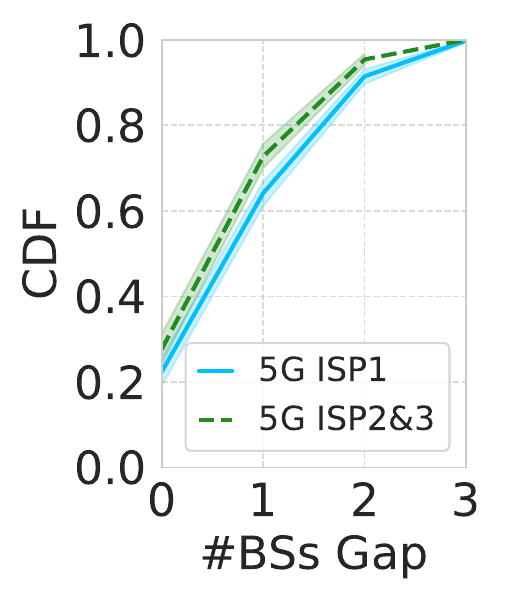}
        \caption{\#BS gap distribution} 
        \label{fig:bs_gap_cdf}
    \end{subfigure}
    \caption{Comparison of SINR and BS imbalance across ISPs.} %
    \label{fig:sinr_bs_gap}
\end{figure}

Channel allocation is one of the most effective approaches to mitigating interference. Although both 4G and 5G offer a limited number of available channels, there exists a significant imbalance in how BSs are distributed across channels. As shown in Table~\ref{tab:band_coverage}, 85\% of 4G cells are concentrated on just 50\% of the available channels. This imbalance leads to uneven interference levels and significantly different SINRs across channels (Figure~\ref{fig:lte_metrics}). To better understand the impact of channel allocation, we analyze how channel allocation for cells in each grid affects SINR, where each grid is of size 100$\times$100 m$^2$ as before. For each grid, we compare two channels by computing the difference in the number of associated cells, defined as $a - b$, and the corresponding difference in average SINR. As shown in Figures~\ref{fig:bs_sinr_gap_isp1} and~\ref{fig:bs_sinr_gap_isp2}, for ISP-1 (n79-f1 vs. n4-f1) and ISP-2\&3 (n78-f2 vs. n78-f1), the average SINR difference reaches the minimum when the number of cells on the two channels is balanced (i.e., $a - b = 0$), while uneven channel allocation results in SINR imbalance and a decrease in overall SINR performance.


\subsection{Time-Frequency Resource Allocation} 
\label{sec:resource_allocation}

UEs may experience inter-cell interference when two cells transmit on the same channel. However, such interference can be mitigated — or even avoided — through fine-grained time-frequency resource coordination between BSs. To investigate whether existing cellular networks implement such coordination between BSs, 
we monitor and analyze the traffic of two inter-BS interfering 4G cells simultaneously using two time-synchronized USRPs. In our measurement, BSs tend to allocate RBs from the lowest frequencies upward, rather than distributing them evenly across the band. As a result, the lower portion of the frequency band consistently experiences higher levels of interference, whereas the upper RBs tend to see relatively less interference. Figure~\ref{fig:lte_prb_heatmap} shows RB usage over a 20-minute trace, where each RB is highlighted for the subframes in which it is used. It can be clearly seen that BSs allocate time-frequency resources in order from lower-indexed RBs to higher-indexed ones. A more detailed usage distribution of RBs is shown in Figure~\ref{fig:PRB_occur_rate}. The first 6 to 10 RBs of OFDM symbols have much higher utilization than other RBs, because they are mainly used to include control information for initial access procedure. Given the frequency-selective nature of wireless channels, allocating RBs based on fixed index preferences makes it difficult to fully utilize spectrum resources. More importantly, neighboring cells tend to collide on the lower-indexed RBs, even when the overall frequency resource utilization is low, significantly degrading the entire network performance. 

When two interfering cells transmit on the same RB of a channel, inter-cell interference inevitably occurs on that RB. We compute the collision probability of RB selection as follows. Suppose that two interfering cells use the RB sets 
$A$ and $B$ at the same time. The collision probability of RBs in $A$ can then be calculated as
\begin{equation*}
    \text{Collision Probability} = \frac{|A \cap B|}{|A|},
\end{equation*}
where $|A|$ is the cardinality of set $A$. The collision probability of RBs in $A$ is shown in Figure~\ref{fig:heatmap}. The diagonal entries represent the collision probabilities when the two cells have similar RB utilization levels. When the utilization is below 20\%, collisions mainly occur for the control traffic in the PDCCH. Even after accounting for the traffic data in the PDSCH, the collision probability remains above 64\%. Such a high collision probability cannot be attributed to periodic RSs in the PDSCH, as they occupy only 5\% of the PDSCH resources. This confirms that the preference for allocating RBs from lower to higher frequencies leads to severe RB collisions and consequently inter-cell interference, even when the resource utilization is low.
\begin{figure}[t]
    \centering
    \hspace{-0.2cm}
    \begin{subfigure}{0.22\textwidth}
        \centering
        \includegraphics[scale=0.32]{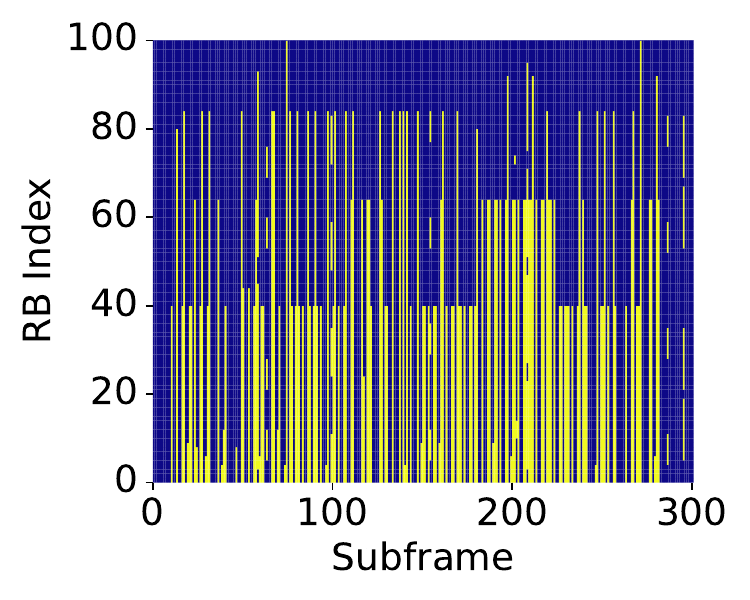}
        \caption{RB usage preference}
        \label{fig:lte_prb_heatmap}
    \end{subfigure}
    \hspace{0.25cm}
    \begin{subfigure}{0.22\textwidth}
        \centering
        \includegraphics[scale=0.32]{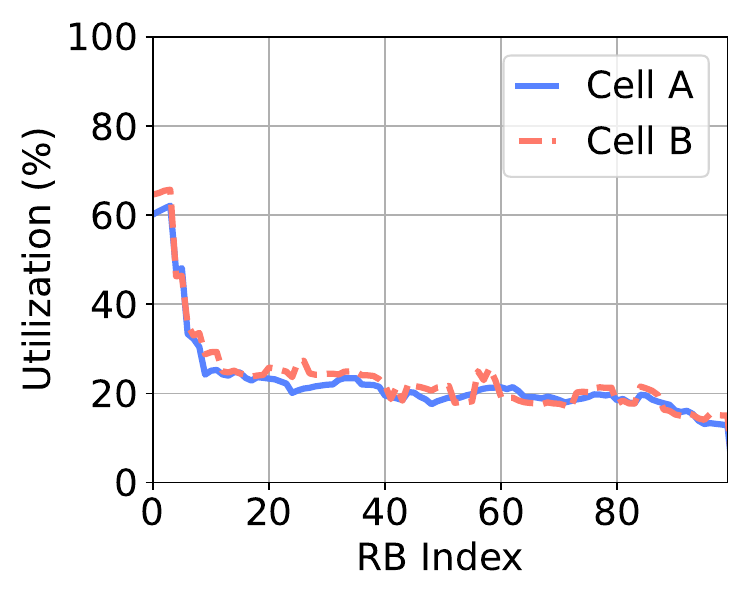}
        \caption{RB utilization distribution}
        \label{fig:PRB_occur_rate}
    \end{subfigure}
    \caption{BSs prioritize allocating resource blocks in ascending order, from lower-indexed RBs to higher-indexed ones.}
    \vspace{-0.5em}
\end{figure}

\begin{figure}[!t]
    \centering
    \includegraphics[scale=0.25]{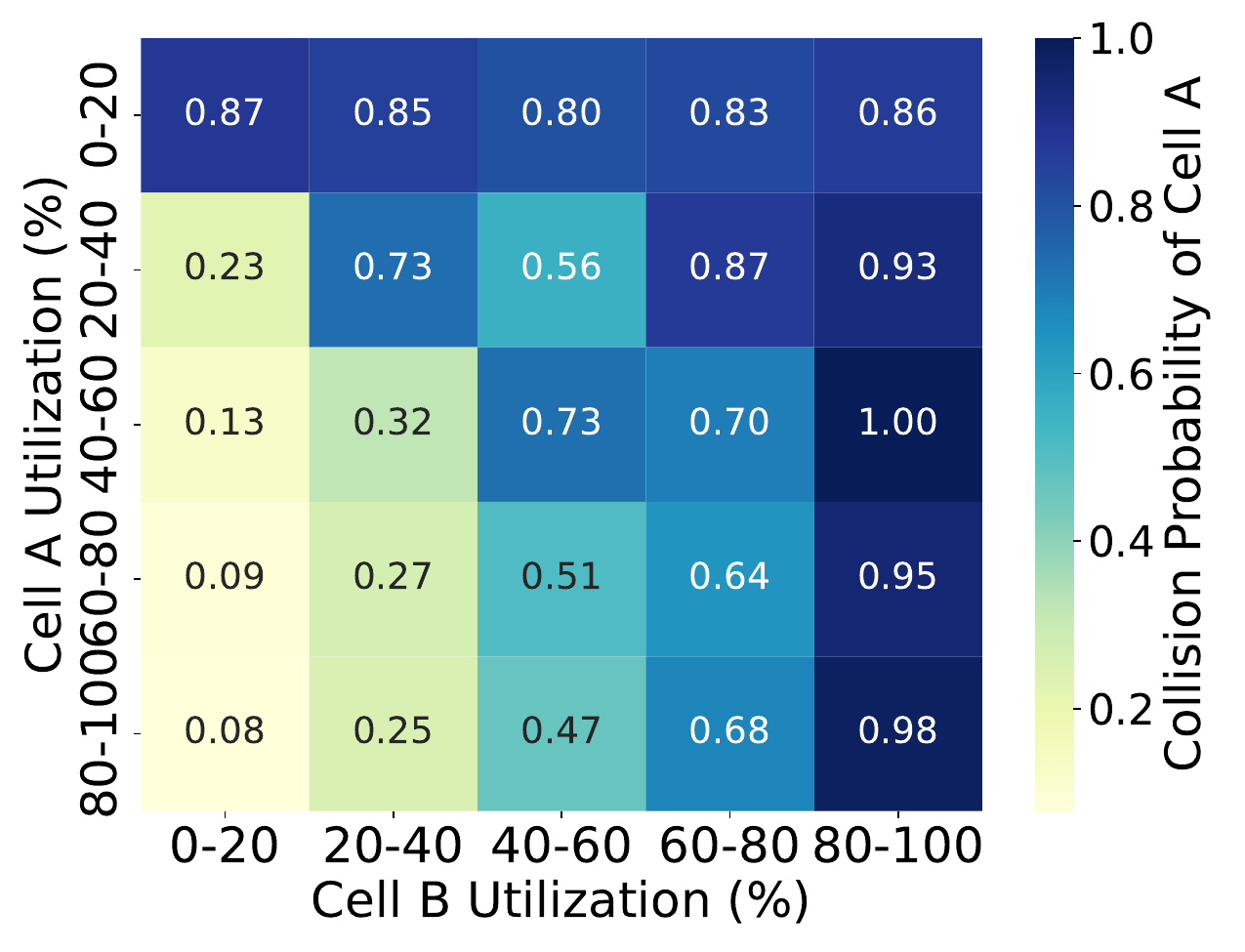}
    \caption{Collisions occur with high probability even when the resource utilization is low.}
    \label{fig:heatmap}
    \vspace{-0.5em}
\end{figure}

\subsection{Configuration: PCI Collision} 

\textbf{Prevalence.} In 4G and 5G networks, each cell is assigned a PCI, which determines the positioning of its reference signals. The CRS in 4G and the DeModulation Reference Signal (DM-RS) in 5G both follow a PCI-dependent positioning. Specifically, the RS position of a cell is selected based on the position index, computed as PCI modulo 
$X$, where $X$ denotes the number of available RS patterns. A PCI collision occurs when two cells share the same RS position index, i.e., RSs are mapped to the same resource elements within a slot, resulting in persistent co-channel interference. Since RSs are used for channel estimation and demodulation, neighboring cells typically adopt different RS positions to reduce mutual interference and improve estimation accuracy. However, our measurements reveal that a majority of 4G and 5G cells have a large number of interfering cells (Figure~\ref{fig:prevalence}), much greater than the number of available RS positions. As a result, PCI collisions occur quite often in existing 4G and 5G networks, degrading the accuracy of channel estimation. 
In our measurements, we detect the PCI collision at each UE location and compute the PCI collision probability as the percentage of locations where a collision is detected. As shown in Figure~\ref{fig:conllision_ratio}, the collision probability increases with the number of interfering cells. In 4G, the collision probability exceeds 90\% when the number of interfering cells reaches 5. Due to the denser deployment, the number of interfering neighbors in 5G is significantly higher than that in 4G, resulting in a 100\% collision probability when the number of interfering cells exceeds 10. When counting collided cells observed by UEs, we find that about 60\% of 4G UEs and 70\% of 5G UEs experience at least one PCI collision (Figure~\ref{fig:collision_count_cdf}).


\begin{figure}[t]
    \centering
    \begin{subfigure}{0.21\textwidth}
        \centering
        \includegraphics[scale=0.3]{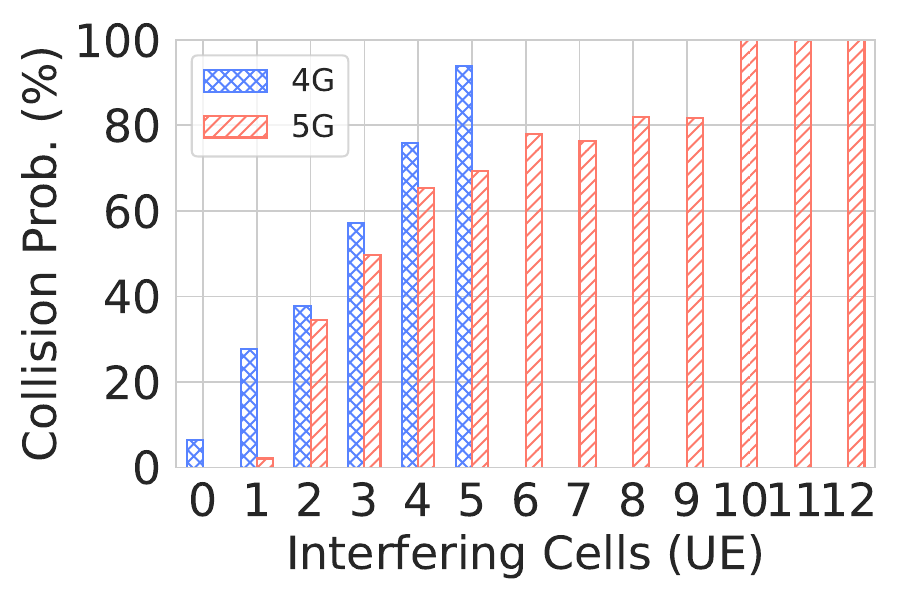}
        \caption{Collision probability}
        \label{fig:conllision_ratio}
    \end{subfigure}
    \hspace{0.2cm}
    \begin{subfigure}{0.23\textwidth}
        \centering
        \includegraphics[scale=0.3]{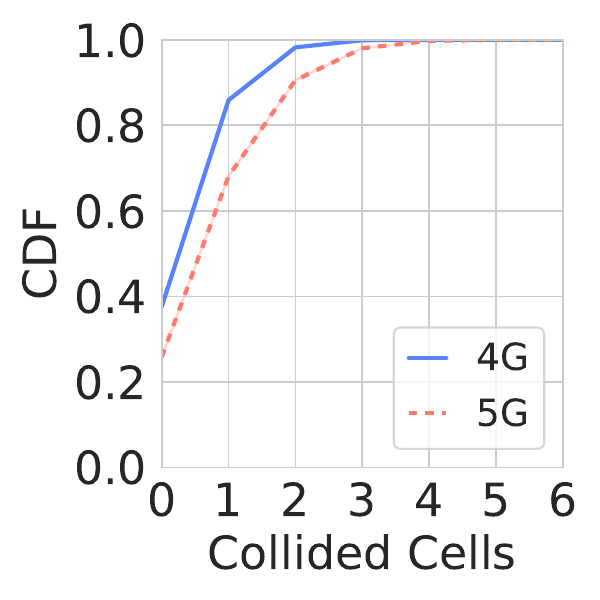}
        \caption{Distribution of collided cells}
        \label{fig:collision_count_cdf}
    \end{subfigure}
    \caption{Prevalence of PCI collisions}
    \vspace{-1em}
\end{figure}

\begin{figure}[t]
    \centering
    \begin{subfigure}{0.23\textwidth}
        \centering
        \includegraphics[scale=0.33]{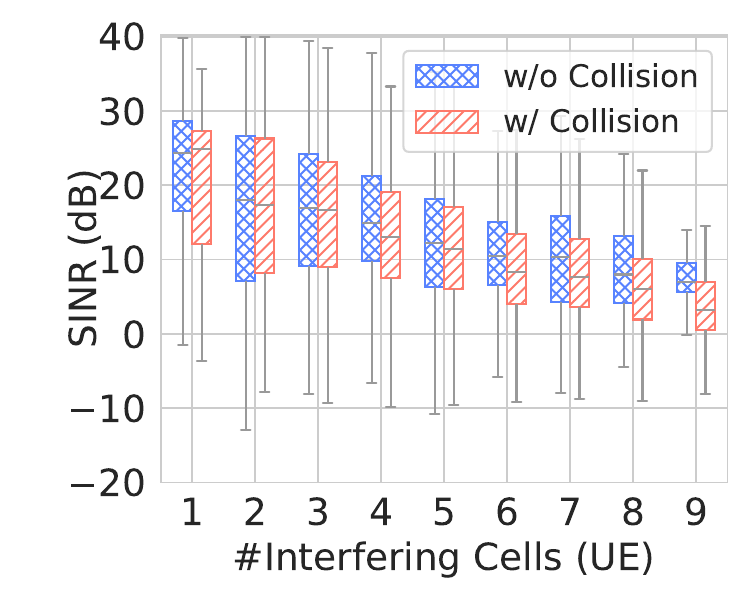}
        \caption{5G PCI collision on SINR}
        \label{fig:5g_collision_sinr}
    \end{subfigure}
    \hfill
    \begin{subfigure}{0.23\textwidth}
        \centering
        \includegraphics[scale=0.33]{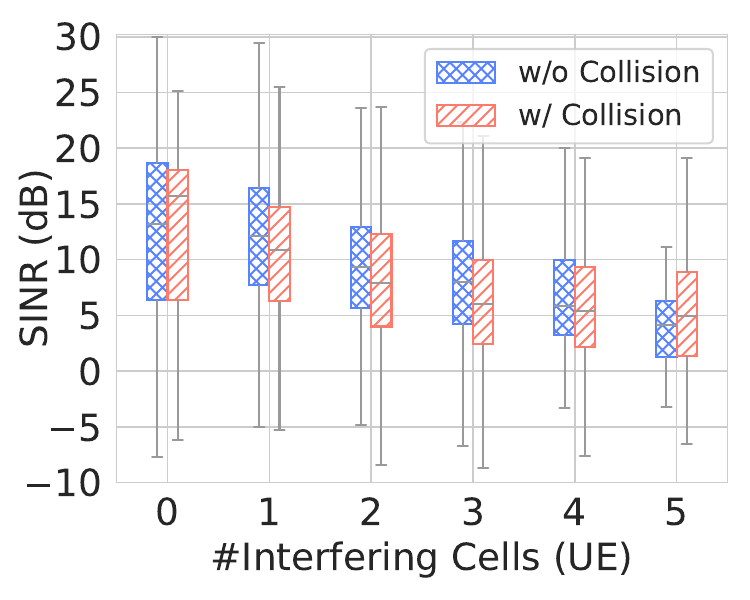}
        \caption{4G PCI collision on SINR}
        \label{fig:lte_collision_sinr}
    \end{subfigure}
    \caption{Impact of PCI collisions on SINR}
    \vspace{-1em}
\end{figure}

\begin{figure*}[t]
    \begin{minipage}[t]{0.3\textwidth}
        \centering
        \includegraphics[scale=0.35]{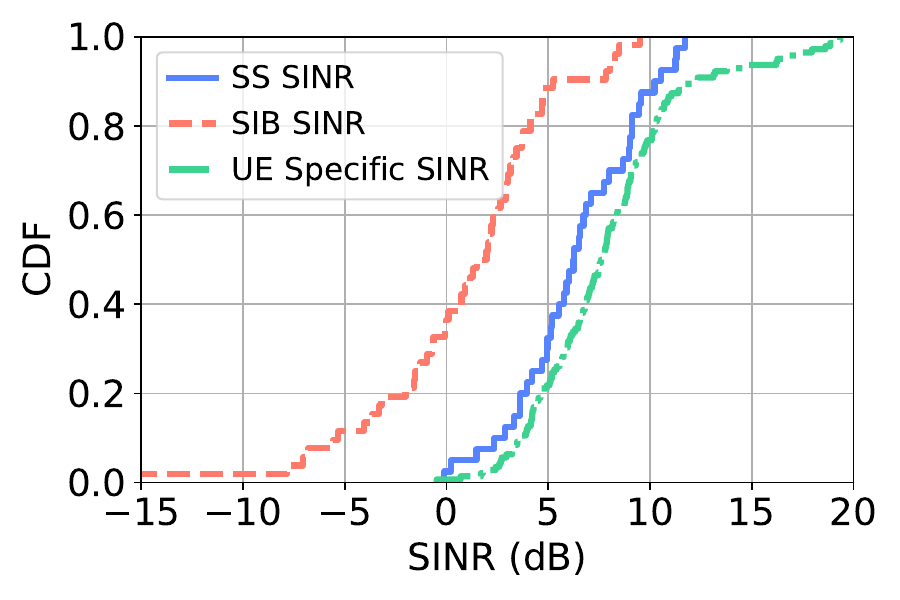}
        \caption{SINR by message types}
        \label{fig:SINR_type_cdf}
    \end{minipage}
    \hfill
    \begin{minipage}[t]{0.6\textwidth}
        \centering
        \includegraphics[scale=0.35]{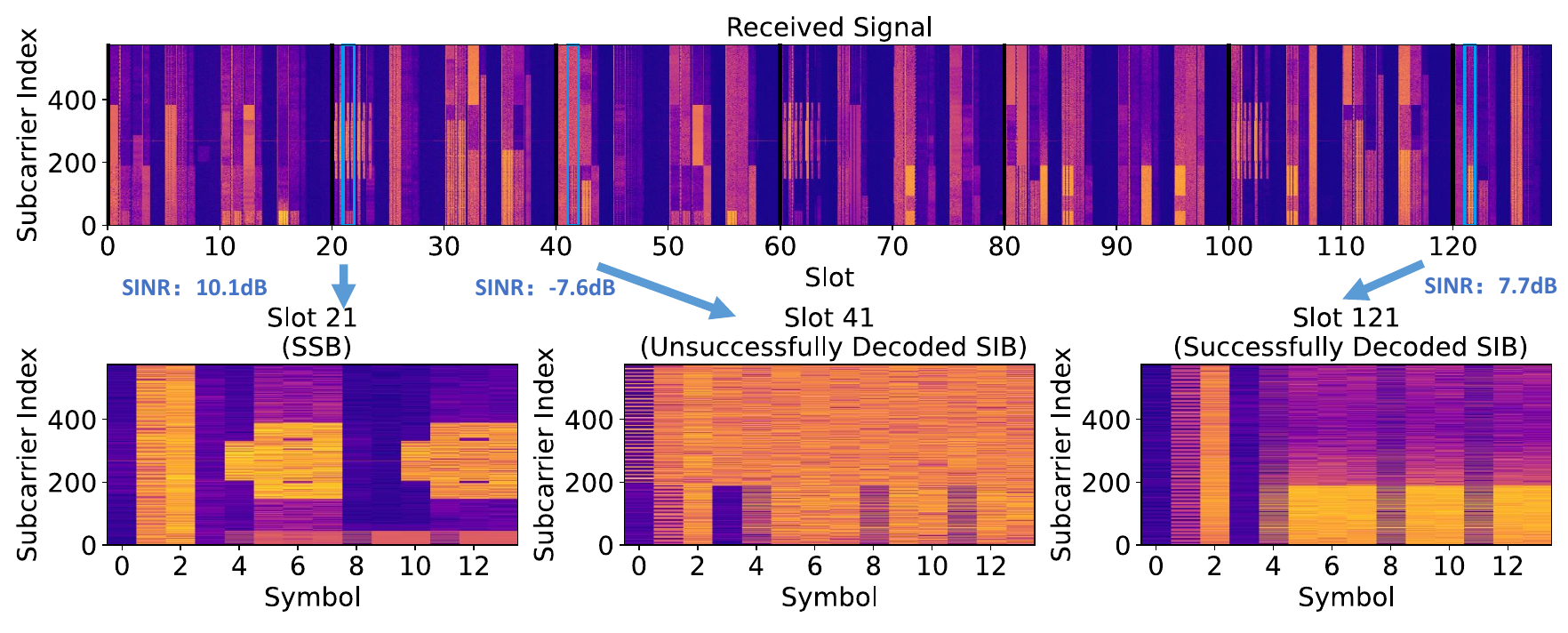}
        \caption{
        An example of unsuccessful SIB decoding
        }
        \label{fig:rxGrid_zoom}
    \end{minipage}
    \hfill
    \vspace{-1em}
\end{figure*}



\textbf{Impact on SINR.} 
We further evaluate the impact of prevalent PCI collisions on SINR. Specifically, we divide UEs in two groups, those affected by PCI collisions and those unaffected, and compare their SINRs under various interference levels. As shown in Figure~\ref{fig:5g_collision_sinr}, UEs with PCI collisions consistently experience lower average SINR than those without across all interference levels in 5G. When the number of interfering cells reaches 6, the SINR gap between the two groups of UEs increases to 4 dB. As previously shown in Figure~\ref{fig:interfering_cells_cdf_by_isp}, more than 10\% of 5G UEs experience at least 6 interfering cells. It is noticeable that the SINR gap between the two UE groups increases with the number of interfering cells. This is because when the number of interfering cells is small, there exists only one or two PCI-collided cells, which may not be strong interferers. However, as the number of interfering cells increases, the likelihood of colliding with strong interfering cells grows, leading to a larger SINR gap between the UE groups. A similar SINR gap is also observed between the two UE groups in 4G. However, since the number of PCI-collided cells in 4G is lower than in 5G (Figure~\ref{fig:collision_count_cdf}), the resulting SINR gap due to PCI collisions is smaller (Figure~\ref{fig:lte_collision_sinr}).

\textbf{Impact on the Signaling Process.}
The process of user access to the network generally consists of several stages, including synchronization, system information acquisition, connection establishment, and data transmission. Since RSs are configured differently for messages at each stage, SINR varies significantly across message types. A critical consequence of this is that the demodulation of messages with low SINR can become a bottleneck in the entire communication process. We consider three message types, i.e., SS, SIB-1, and UE-specific data, and compute the SINR for messages at all UE locations. As shown in Figure~\ref{fig:SINR_type_cdf}, SIB consistently experiences much lower SINR than SS, which in turn has lower SINR than user-specific data. The root cause is the limited number of DMRS placement options available for SIB. In our measurements, only two options are observed for SIB, compared to four for SS. This reduced flexibility increases the likelihood of collisions and leads to degraded SINR for SIB. In contrast, UE-specific data is randomly scheduled across the time-frequency grid, making interference less persistent and often absent, resulting in the highest SINR among the three. Since SS SINR is commonly used to measure UE channel quality, the SINR gap between message types leads to two key issues: 1) SIB would inadvertently suffer from extremely low SINR and 2) the SINR is underestimated for user-specific data.
Figure~\ref{fig:rxGrid_zoom} presents a real-world example where a SIB fails to decode due to low SINR, with PCI collisions being a primary contributing factor.
\section{Related Work}
\textbf{4G/5G Measurements.} Extensive studies \cite{5g_midband,5g_2020,5g_2021,5g_mmwave} have measured 4G/5G networks from various aspects, such as coverage, latency, throughput, energy efficiency, and application performance. However, none of these works characterize the impact of inter-cell interference on network performance. A limited body of work has addressed topics related to inter-cell interference, but no in-depth analysis has been conducted. In~\cite{ca24,ca23}, CA strategies have been designed in multi-cell scenarios, but interference is not considered in their proposed solutions. Additionally, \cite{lteye}~discusses the significant difference between SINR and SNR, yet fails to analyze the root causes. \cite{5g_reliability} acknowledges that interference may increase the cellular failure rate, but do not conduct a detailed investigation. To the best of our knowledge, there exists no comprehensive study on inter-cell interference in 4G and 5G networks. Our work fills this gap by conducting the first measurement study of inter-cell interference.

\textbf{Inter-Cell Interference Management.} 
Interference management in 4G and 5G networks relies heavily on RSs for interference measurement~\cite{csi_rs}. Extensive works \cite{pilot_v2x,pilot_numerology,pilot_urllc} have explored the design of RS patterns to meet various requirements of services in the 5G network. Building on the interference measurement, several studies \cite{interference_management_1,interference_management_2,interference_management_3} have proposed interference management schemes. While interference is measurable, most existing systems fail to account for inter-cell interference. For instance, \cite {dynamic_tdd} focuses solely on dynamic TDD configurations, which can lead to inter-link interference due to the uplink and downlink transmission interactions, but does not address the resulting interference. Similarly, studies on network slicing and scheduling strategies \cite{5g_slicing,slicing_24} do not consider how multi-cell coordination can help reduce interference. Our measurement study highlights the importance of considering inter-cell interference in network design.

\textbf{Network Deployment and Configuration.}
Some studies~\cite{deployment_miami,5g_deployment,5g_beam,deployment_bandwidth} show that different BSs exhibit varying characteristics depending on deployment and configuration factors such as access technology, frequency bands, and ISPs. For example, \cite{5g_deployment} compares different system information configuration schemes used by various ISPs, showing that while each provider may have its own allocation strategy, they tend to apply the same SIB and SSB allocation strategy across all of their managed BSs. These deployment and configuration factors can have potential impact on inter-cell interference, though their impacts are not evaluated in existing studies. Our study conducts an in-depth analysis of deployment and configuration issues to mitigate inter-cell interference.

\section{Discussion}

\textbf{Application-Level Performance.}
Many existing studies\cite{deployment_bandwidth,5g_midband,ca24} have investigated the impact of SINR and upper-layer protocols on application  performance. In our work, the focus is primarily on the physical layer performance, using SINR and physical-layer throughput to reflect the impact of inter-cell interference. Moreover, during practical BS deployment and optimization, metrics such as interference and received signal strength provide a more accurate and direct representation of the underlying network performance. With physical-layer performance ensured, further optimization can then be carried out at upper layers.

\textbf{Limitations in Measurement Methodology.}
Unlike 4G, 5G does not transmit periodic RSs across all frequencies. Aside from system information such as SSB and SIB, user-specific data in 5G is distributed sparsely and randomly across the spectrum. Characterizing RB-level channel characteristics in a 100 MHz 5G system is therefore a challenging task, subject to future research.

\section{Conclusion}
We present a first comprehensive measurement study of inter-cell interference in operational 4G and 5G networks. Our results reveal that inter-cell interference is prevalent in both 4G and 5G networks and significantly impacts signal quality and throughput. Using a carefully designed method for fine-grained channel estimation, we characterize inter-cell interference under frequency-selective fading, exposing limitations in the existing handover and resource allocation strategies. Additionally, we analyze the root causes of inter-cell interference from four perspectives and find that the current resource allocation scheme results in high probability of collisions in time-frequency resource usage between base stations. Moreover, the configuration of reference signals causes certain types of message exchanges to become bottlenecks in the communication process. Our findings provide insights for future inter-cell interference management.

\appendices
\section{Proof of Theorem \ref{theorem:bennet}}
\label{appendix:theorem}

Assume that the interference and noise are random variables with a zero mean, i.e., $\mathbb{E}[X_i] = \mathbb{E}[N] = \mathbb{E}[\frac{\Delta_i}{h_1}] = 0$. Let $\tilde{\Delta}_i = \frac{\Delta_i}{h_1}$. The variance can be estimated as
\begin{equation*}
    \Var(|\tilde{\Delta}_i|^2) = \mathbb{E}[|\tilde{\Delta}_i|^2] = \frac{\mathbb{E}\left[\sum_{j=2}^K\left|h_jX_j\right|^2 + |N|^2\right]}{\mathbb{E}[|h_1X_1|^2]} = \frac{1}{SINR}
\end{equation*}
Since noise is negligible compared to interference, using the AM-QM inequality, we can bound $|\tilde{\Delta}_i|$ as
\begin{equation*}
    |\tilde{\Delta}_i| \leq \max\left(\sum_{j=2}^K\frac{|h_jX_j|}{|h_1X_1|} \right) \leq b 
\end{equation*}
where $b = \max\left(\sqrt{\frac{(K-1)\sum\limits_{j=2}^K |h_jX_j|^2}{|h_1X_1|^2}}\right).$Using the Bennett's Inequality \cite{Bennett}, we have
\begin{equation*}
    \mathbb{P}\left[\frac{1}{m}\left|\sum_{i=1}^m\tilde{\Delta}_i\right| \geq \delta \right] \leq \text{exp}\left(-\frac{m\sigma^2}{b} h\left(\frac{b\delta}{\sigma^2}\right)\right) \leq \epsilon,
\end{equation*}
where $\sigma^2 = \Var(|\tilde{\Delta}_i|^2)$.

\section{Inter-cell interference in urban and rural areas}
\label{sec:urban}

To generalize our findings for 4G and 5G campus networks, we conducted real-world measurements in both urban and rural areas (Figure~\ref{fig:urban_rural}) using the same measurement methodology and tools introduced in \S\ref{sec:methodology}. We confirmed that all the key observations for inter-cell interference on campus also hold in urban and rural areas as follows. (1) From the user perspective, inter-cell interference is prevalent but its severity varies by scenario (urban $>$ campus $>$ rural, Figure~\ref{fig:interfering_cells_cdf_by_environment}). (2) Low-frequency region is prioritized for resource allocation (Figure \ref{fig:low frequency priority}). (3) Interference power increases as the number of interfering cells per UE grows, resulting in a decrease in SINR (Figures~\ref{fig:4g order} and \ref{fig:4g sinr}). (4) PCI collisions exist in both 4G and 5G (Figure \ref{fig:collision}). Due to page limit, Figures \ref{fig:4g order}-\ref{fig:collision} present only the results in the urban area, but observations (3) and (4) also hold in the rural area.

\begin{figure}[h]
\centering
\begin{subfigure}{0.14\textwidth}
    \includegraphics[width=\textwidth]{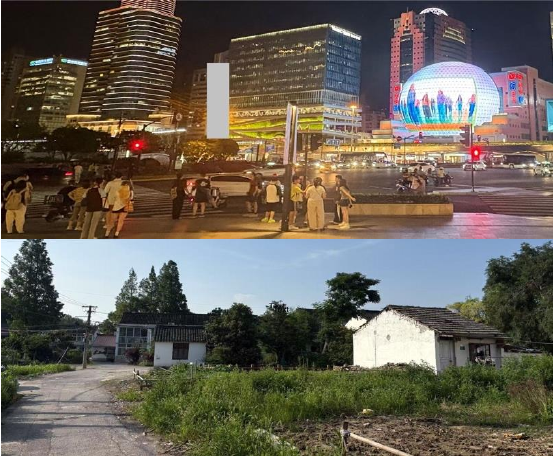}
    \caption{Urban and rural}
    \label{fig:urban_rural}
\end{subfigure}
\begin{subfigure}{0.16\textwidth}
    \includegraphics[width=\textwidth]{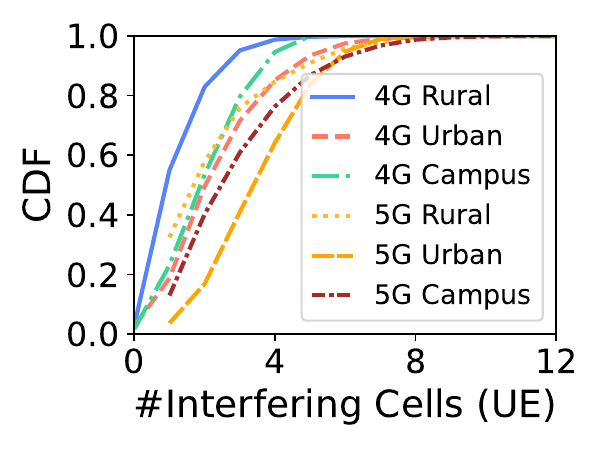}
    \caption{\#Interfering cells}
    \label{fig:interfering_cells_cdf_by_environment}
\end{subfigure}
\hspace{-0.2cm}
\begin{subfigure}{0.16\textwidth}
    \includegraphics[width=\textwidth]{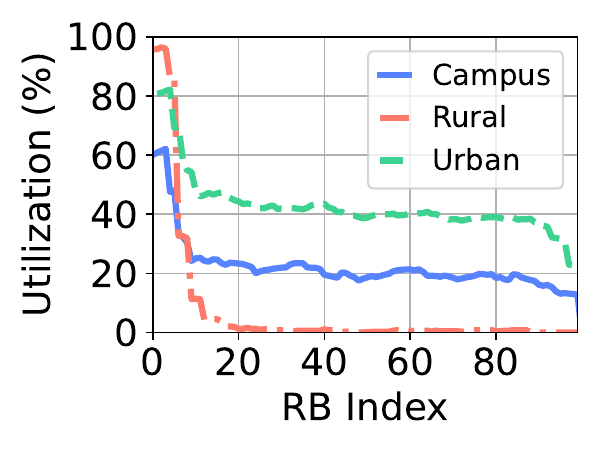}
    \caption{RB utilization}
    \label{fig:low frequency priority}
\end{subfigure}
\begin{subfigure}{0.18\textwidth}
    \includegraphics[width=\textwidth]{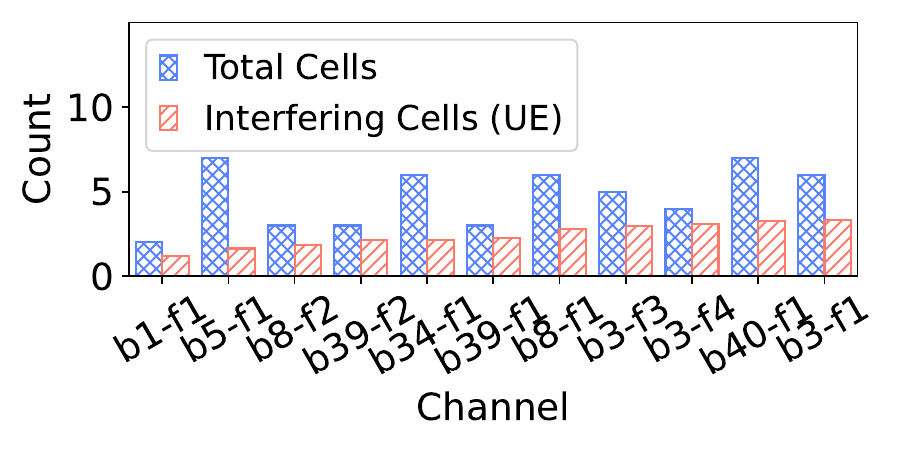}
    \caption{4G ordered channels}
    \label{fig:4g order}
\end{subfigure}
\hspace{-0.3cm}
\begin{subfigure}{0.18\textwidth}
    \includegraphics[width=\textwidth]{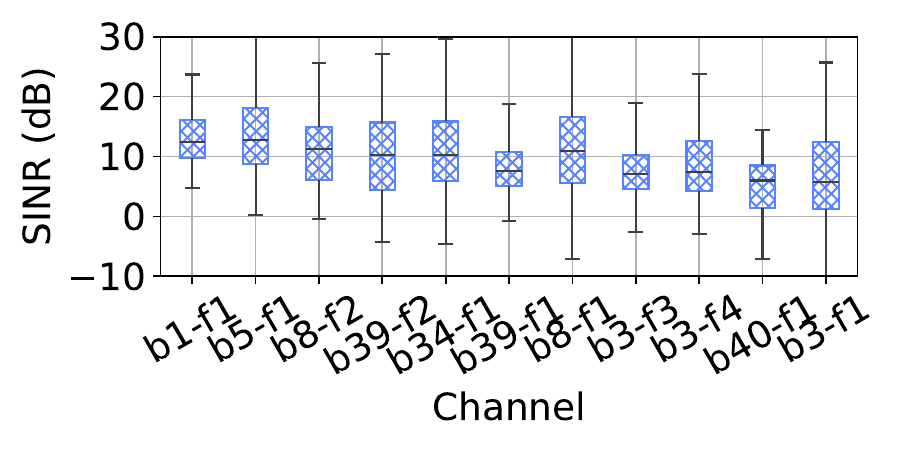}
    \caption{4G SINR}
    \label{fig:4g sinr}
\end{subfigure}
\begin{subfigure}{0.12\textwidth}
    \includegraphics[width=\textwidth]{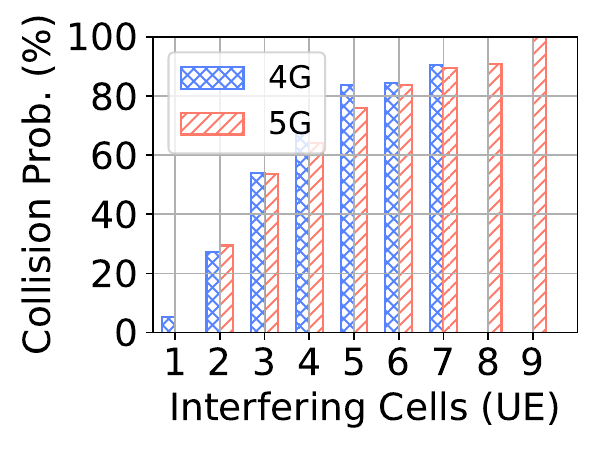}
    \caption{PCI Collision}
    \label{fig:collision}
\end{subfigure}
\caption{(a)-(c): urban and rural areas; (d)-(f): urban area}
\vspace{-1em}
\end{figure}

\bibliographystyle{IEEEtran}
\bibliography{reference}

\end{document}